\documentclass[english,onecolumn,draftcls]{IEEEtran}

\usepackage{hyperref}
\usepackage[T1]{fontenc}
\usepackage[latin9]{inputenc}
\usepackage{amsthm}
\usepackage{amsmath}
\usepackage{amssymb}
\usepackage{graphicx}
\usepackage{dsfont}
\usepackage{epstopdf}
\usepackage{color}

\newcommand{\overbar}[1]{\mkern 1.25mu\overline{\mkern-1.25mu#1\mkern-0.25mu}\mkern 0.25mu}

\newcommand{\mywidehat}[1]{\mkern 1.5mu\widehat{\mkern-1.5mu#1\mkern-0mu}\mkern 0mu}

\DeclareMathOperator*{\argmin}{arg\,min}

\newcommand{\openone}{\mathds{1}}

\newcommand{\blambda}{\boldsymbol{\lambda}}
\newcommand{\domin}{D_{0,\min}}
\newcommand{\doprod}{D_{0,{\rm prod}}}
\newcommand{\diprod}{D_{1,{\rm prod}}}

\newcommand{\Xtypinf}{\mathcal{X}^{\infty}_{\rm typ}}

\newcommand{\sign}{{\rm sign}}

\newcommand{\Rbar}{\overbar{R}}
\newcommand{\Dbar}{\overbar{D}}

\newcommand{\Xvhat}{\mywidehat{\boldsymbol{X}}}
\newcommand{\xvhat}{\mywidehat{\boldsymbol{x}}}
\newcommand{\Xhat}{\mywidehat{X}}
\newcommand{\xhat}{\mywidehat{x}}
\newcommand{\Xchat}{\mywidehat{\mathcal{X}}}

\newcommand{\Dibar}{\overbar{D}_1}
\newcommand{\Dobar}{\overbar{D}_0}

\newcommand{\Pctilde}{\widetilde{\mathcal{P}}}

\newcommand{\Phat}{\mywidehat{P}}

\newcommand{\Ptilde}{\widetilde{P}}

\newcommand{\uv}{\boldsymbol{u}}

\newcommand{\Uv}{\boldsymbol{U}}

\newcommand{\xv}{\boldsymbol{x}}

\newcommand{\Xv}{\boldsymbol{X}}

\newcommand{\Ac}{\mathcal{A}}

\newcommand{\Cc}{\mathcal{C}}

\newcommand{\Ec}{\mathcal{E}}

\newcommand{\Ic}{\mathcal{I}}

\newcommand{\Pc}{\mathcal{P}}
\newcommand{\Sc}{\mathcal{S}}
\newcommand{\Tc}{\mathcal{T}}
\newcommand{\Uc}{\mathcal{U}}

\newcommand{\Xc}{\mathcal{X}}

\newcommand{\EE}{\mathbb{E}}
\newcommand{\PP}{\mathbb{P}}
\newcommand{\RR}{\mathbb{R}}

\newcommand{\bone}{\boldsymbol{1}}

\makeatletter

\theoremstyle{plain}
\newtheorem{thm}{\protect\theoremname}
\theoremstyle{plain}
\newtheorem{lem}{\protect\lemmaname}
\theoremstyle{plain}

\theoremstyle{plain}
\newtheorem{defn}{\protect\definitionname}
\theoremstyle{plain}

\theoremstyle{plain}
\newtheorem{cor}{\protect\corollaryname}

\usepackage{flushend}
\usepackage[noadjust]{cite}
\usepackage{enumerate}
\usepackage{amsmath}
\usepackage{graphicx}
\usepackage{amssymb}
\usepackage{subfigure}
\usepackage{booktabs}
\usepackage{amsfonts}
\usepackage{amssymb}
\usepackage{bbm}

\makeatother

\usepackage{babel}
\providecommand{\lemmaname}{Lemma} 
\providecommand{\theoremname}{Theorem}
\providecommand{\propositionname}{Proposition}
\providecommand{\definitionname}{Definition}
\providecommand{\assumpname}{Assumption}
\providecommand{\corollaryname}{Corollary}

\begin{document}
\author{Millen Kanabar and Jonathan Scarlett \thanks{M.~Kanabar is with the Department of Electrical Engineering, IIT Bombay.  J.~Scarlett is with the  Department of Computer Science, Department of Mathematics, and Institute of Data Science, National University of Singapore (NUS).  e-mail: \url{millen@ee.iitb.ac.in}; \url{scarlett@comp.nus.edu.sg}}
\thanks{This work was presented in part at the IEEE International Symposium on Information Theory (ISIT), 2022.}}

\long\def\symbolfootnote[#1]#2{\begingroup\def\thefootnote{\fnsymbol{footnote}}\footnote[#1]{#2}\endgroup}

\title{Mismatched Rate-Distortion Theory: Ensembles, Bounds, and General Alphabets}
\maketitle

\begin{abstract}
    In this paper, we consider the mismatched rate-distortion problem, in which the encoding is done using a codebook, and the encoder chooses the minimum-distortion codeword according to a mismatched distortion function that differs from the true one.  For the case of discrete memoryless sources, we establish achievable rate-distortion bounds using multi-user coding techniques, namely, superposition coding and expurgated parallel coding.  We study examples where these attain the matched rate-distortion trade-off but a standard ensemble with independent codewords fails to do so.  On the other hand, in contrast with the channel coding counterpart, we show that there are cases where structured random codebooks can perform worse than their unstructured counterparts.  In addition, in view of the difficulties in adapting the existing and above-mentioned results to general alphabets, we consider a simpler i.i.d.~random coding ensemble, and establish its achievable rate-distortion bounds for general alphabets.
\end{abstract}

\vspace*{-1ex}
\section{Introduction} \label{sec:INTRO}
\vspace*{-0.5ex}

Rate-distortion theory is one of the most fundamental topics in information theory, and since its introduction \cite{Sha59a}, an extensive and diverse range of rate-distortion settings have been considered.  In this paper, we are interested in a setting studied by Lapidoth \cite{Lap97}, in which the encoding is {\em mismatched}, i.e., it is done with respect to an ``incorrect'' distortion measure $d_0$ that may be different from the true measure $d_1$.  
As discussed in \cite{Lap97,Sca20}, this problem can arise when the true distortion measure is not known when designing the encoder/decoder, or when the optimal encoder is ruled out due to implementation constraints (e.g., finite-precision arithmetic).  More broadly, the problem is fundamental in nature, and may provide useful tools for other problems involving compression and/or imperfect encoding rules, e.g. estimation from compressed data \cite{Zha88,Kip21}.

The mismatched rate-distortion problem serves as a natural counterpart to mismatched decoding in channel coding \cite{Sca20}.  However, different aspects turn out to be much easier in one problem than the other.  For instance, an important conjecture on the tightness of a multi-letter achievable rate \cite{Csi95,Som15b} remains open in channel coding, but its analog was resolved in the rate-distortion setting \cite{Lap97}.  On the other hand, general alphabets have been studied in detail in channel coding \cite{Gan00,Sca16a}, but appear to be more challenging in the rate-distortion problem.

In this paper, we study the following lesser-understood aspects of the mismatched rate-distortion problem: (i) multi-user coding techniques in the case of discrete memoryless sources, and (ii) regular random coding techniques in the case of general alphabets.  These aspects have been studied extensively in the mismatched channel coding problem \cite{Lap96,Sca16a,Som15}, and their counterparts in mismatched rate-distortion theory were included in a list of open problems in \cite[Ch.~10]{Sca20}.  Regarding multi-user coding techniques, a simple example consisting of parallel sources in \cite{Lap97} demonstrated that such techniques can be beneficial, but general results have remained unexplored.  Regarding general alphabets, we are not aware of any existing results for the setup we consider, but this direction has been explored in related rate-distortion settings with only a single distortion function (e.g., see \cite{Dem02,Kon06}).

Briefly, our main contributions are as follows:
\begin{itemize}
    \item For discrete memoryless sources, we derive two new achievability results on the mismatched distortion-rate function based on superposition coding and expurgated parallel coding.
    \item We demonstrate that, in contrast to the channel coding problem, introducing structure into the codebook (i.e., statistical dependencies between the random codewords) can be detrimental to the rate-distortion trade-off.
    \item On the positive side, we explore examples where multi-user techniques provide a strict improvement compared to the ensemble with independent codewords, one of which is conceptually very different from the known parallel source example.
    \item In view of difficulties in generalizing the existing and above-mentioned achievability results to general alphabets, we study the simpler i.i.d.~random coding ensemble, and establish an achievable rate that extends to general alphabets.
\end{itemize}

\subsection{Problem Setup} 

Consider a discrete memoryless source (DMS) where each symbol is drawn independently from a source distribution $ \Pi_X(x) $ over a finite alphabet $\Xc$. The distribution of the resulting $n$-letter sequence $\Xv$ is then given by $\Pi_X^n(\xv) = \prod_{i=1}^{n}\Pi_X(x_i)$. The encoder maps $\xv$ to an index $m \in \{1, 2, \dots, M\}$ for some integer $M$, and the rate of the encoder is given by $R = \frac{1}{n}\log M$.

The choice of $m$ is dictated by a codebook $ \mathcal{C} = \{\xvhat^{(1)}, \dots, \xvhat^{(M)}\} $, where each codeword lies in $\Xchat^n$ for some reconstruction alphabet $\Xchat$. We assume that the encoder chooses, for some non-negative encoding metric $d_0: \Xc\times \Xchat \rightarrow \RR_{\ge 0}$,
\begin{align}
    m = \argmin_{1 \leq j \leq M} d^n_0(\xv, \xvhat^{(j)}), \quad d^n_0(\xv, \xvhat) = \sum_{i=1}^{n}d_0(x_i, \xhat_i). \label{eq:enc_rule}
\end{align}
The distortion incurred is given by $ d_1^n(\xv, \xvhat) = \sum_{i=1}^n d_1(x_i, \xhat_i) $ for some true non-negative distortion measure $d_1 : \Xc\times \Xchat \rightarrow \RR_{\ge 0}$.  If $d_0 = d_1$ then this is a standard rate-distortion problem, but otherwise, the encoding is said to be {\em mismatched}.

The choice of tie-breaking strategy in \eqref{eq:enc_rule} can sometimes be significant (see Appendix \ref{app:example_tie} for an example), and we will consider the following possible choices:
\begin{itemize}
    \item Following \cite{Lap97}, the {\em pessimistic strategy} assumes that any ties in \eqref{eq:enc_rule} always lead to the worst-case value of $d_1^n$ with respect to those tied codewords.
    \item Alternatively, ties may be broken uniformly at random.
    \item For convenience, in one of our results, we will assume that ties are always broken by choosing the tied codeword with the lowest index in $\{1,\dotsc,M\}$. 
\end{itemize}
Except where stated otherwise, the achievability results that we state hold regardless of the tie-breaking strategy, i.e., even under pessimistic tie-breaking.  The focus on the pessimistic rule in \cite{Lap97} was primarily for the purpose of proving a multi-letter converse.

\begin{defn}[Mismatched distortion-rate and rate-distortion functions]
    A rate-distortion pair $ (R, D_1) $ is said to be achievable if, for any $\delta > 0$, there exists a sequence of codebooks indexed by n with $ M \leq e^{n(R+\delta)} $ codewords such that 
    \begin{align}
        \mathbb{P}\left[\frac{1}{n}d_1^n (\Xv, \Xvhat) \geq D_1 + \delta\right] \leq \delta
    \end{align}for sufficiently large $ n $, where $ \Xv \sim \Pi_X^n $ and $ \Xvhat $ is the resulting estimate via \eqref{eq:enc_rule}. The rate-distortion function $ R^*(D_1) $ is defined as the smallest rate such that $ (R, D_1) $ is achievable. The distortion-rate function $D_1^*(R)$ is defined as the smallest distortion $ D_1 $ such that $ (R, D_1) $ is achievable.
\end{defn}

\subsection{Related Work}

Our work is most closely related to that of Lapidoth \cite{Lap97}, whose main achievability result for constant-composition (or near-constant composition) coding is stated below in Lemma \ref{lem:existing}.\footnote{More precisely, our setup is a slightly simplified version of that in \cite{Lap97}; in our setup, the decoder must output the selected codeword directly, whereas in \cite{Lap97}, additional post-processing is allowed.  As noted in \cite{Sca20}, the analysis and results of \cite{Lap97} still apply with only minor changes.}  It was also shown in \cite{Lap97} that this result can be improved by applying it to the the product source ($\pi_X^2$, $\pi_X^3$, and so on) and that the resulting achievable distortion has a matching converse in the limit of an increasing product.  We will also build on two example sources studied in \cite{Lap97}.

To our knowledge, follow-up works of \cite{Lap97} in settings with multiple distortion measures have remained fairly limited, at least in aspects that are of direct relevance to our work.  Related problems were studied for joint source-channel coding and strategic communication in \cite{Aky21,Le21}.  Another interesting related setup was recently considered in \cite{Mah21}, in which the distortion measure is unknown when designing the codebook, but is revealed to the encoder at runtime.

Another prominent form of mismatch in the literature is that in which there is only a single distortion measure, but the codebook is designed for the wrong source or is required to work for multiple different sources  \cite{Sak70,Sak69,Lap97,Dem02,Kon06,Yan98,Yan99,Gra75,Gra03}.  We will make use of several useful auxiliary results from \cite{Dem02,Kon06} that fall in this category, even though we consider a different setup with two distortion measures.

When it comes to continuous sources, perhaps the most well-known result concerns the performance of Gaussian codes applied to non-Gaussian sources with the standard squared-distance distortion measure \cite{Sak69,Sak70,Lap97}.  To our knowledge, the setup we consider has not been considered previously with continuous (or more general) alphabets.

Another line of works has studied universal codes for rate-distortion theory \cite{Dem02,Kon06,Ste93,Zha96,Sak69,Sak70}, with the goal of finding codes that work well simultaneously for a large class of sources.  In this setting, the encoder and decoder may be designed specifically to achieve this goal, in contrast with our study of mismatched encoding.   

Finally, our problem setup serves as a natural counterpart to mismatched decoding in channel coding, for which multi-user ensembles were studied in \cite{Lap96,Sca16a,Som15}, and general alphabets were studied in \cite{Kap93,Gan00,Sca16a,ScarlettThesis}, among others.  See \cite{Sca20} for a survey of this topic.

\subsection{Notation}

Our analysis will be based on the method of types \cite{Csi98}.  For a given type $P_U$ (i.e., a possible empirical distribution with block length $n$), we let $\Tc^n(P_U)$ denote the set of sequences $\uv$ whose empirical distribution is $P_U$, and similarly for joint types.  For a fixed sequence $\uv$ whose empirical distribution matches the marginal of some joint type $P_{UX}$, we let $\Tc_{\uv}^n(P_{UX})$ denote the set of all $\xv$ sequences such that $(\uv,\xv)$ has joint empirical distribution $P_{UX}$.  The marginals of a joint distribution $P_{UX}$ are written as $P_U$ and $P_X$.

The set of distributions on a given alphabet is denoted by $\Pc(\cdot)$, and the set of types is denoted by $\Pc_n(\cdot)$.  Logarithms and information quantities have base $e$ except where stated otherwise.  We make use of the binary entropy function $H_2(a)$ (i.e., the entropy with probabilities $(a,1-a)$) and the ternary entropy function $H_3(a,b)$ (i.e., the entropy with probabilities $(a,b,1-a-b)$).

\vspace*{-1ex}
\section{Random Coding with Independent Codewords}
\vspace*{-0.5ex}

As a precursor to our main results, we overview the achievable distortion-rate function derived in \cite{Lap97} based on constant-composition random coding, and present a closely-related result for i.i.d.~random coding.

\subsection{Constant-Composition Ensemble} \label{sec:cc}

Consider an auxiliary distribution $Q_{\Xhat} \in \Pc(\hat{\Xc})$, and let $Q_{\Xhat, n} \in \Pc_n(\hat{\Xc})$ be a type with the same support as $Q_{\Xhat}$ such that $ \|Q_{\Xhat, n} - Q_{\Xhat}\|_\infty \leq \frac{1}{n}$.  In the constant-composition ensemble, we draw each codeword independently from
\begin{align} \label{ccCodingDist}
    P_{\Xvhat}(\xvhat) &= \frac{1}{|\Tc^n(Q_{\Xhat, n})|}\openone\{\xvhat \in\Tc^n(Q_{\Xhat, n}) \},
\end{align}where $ \Tc^n(Q_{\Xhat, n}) $ is the type class corresponding to $Q_{\Xhat, n}$.
The following lemma is stated in \cite[Lemma 4.3]{Sca20} based on the analysis of \cite{Lap97}.

\begin{lem}[Occurrences of joint types \cite{Lap97,Sca20}] \label{lem:occurence_of_types}
    Fix $\Pi_X$, $Q_{\Xhat}$, and $R$, and consider a random codebook $\Cc_n = \{\Xv^{(1)}, \dots, \Xv^{(M)}\}$ with $M = \lfloor e^{nR} \rfloor$ codewords of length $n$ drawn independently from $ P_{\Xvhat} $ in \eqref{ccCodingDist}. For any $ \delta>0 $, conditioned on any $ \Xv = \xv $ with $\xv \in \Tc(P_{X,n})$ for some type $P_{X,n}$, the following statements hold with probability approaching one as $n\rightarrow\infty$:
    \begin{enumerate}
        \item For all $\xvhat \in \Cc_n $, if $ (\xv, \xvhat) \in \Tc^n (\Ptilde_{X\Xhat}) $ for some $ \Ptilde_{X\Xhat} \in \Pc_n(\Xc \times \hat{\Xc}) $, then it must hold that $ \Ptilde_X = P_{X,n} $ and $ \Ptilde_{\Xhat} = Q_{\Xhat, n} $.
        \item Defining the set
        \begin{align}
            \Sc_{n,\delta}^{\supseteq} = \left\{\Ptilde_{X\Xhat} \in \Pc_n(\Xc \times \hat{\Xc}) \,:\, \Ptilde_X = P_{X,n}, \Ptilde_{\Xhat} = Q_{\Xhat, n}, I_{\Ptilde}(X, \Xhat) \leq R + \delta \right\}, \label{eq:S_super}
        \end{align}
        all joint types $\Ptilde_{X\Xhat}$ induced by the codewords (i.e., $(\xv, \xvhat) \in \Tc^n( \Ptilde_{X\Xhat})$ for some codeword $\xvhat$) must satisfy $ P_{X\Xhat}\in \Sc_n^{\supseteq}(P_{X,n}) $.
        \item For any joint type $\Ptilde_{X\Xhat}$ in the set
        \begin{align}
            \Sc_{n,\delta}^{\subseteq} = \left\{\Ptilde_{X\Xhat} \in \Pc_n(\Xc \times \hat{\Xc}) \,:\, \Ptilde_X = P_{X,n}, \Ptilde_{\Xhat} = Q_{\Xhat, n}, I_{\Ptilde}(X, \Xhat) \leq R - \delta  \right\},  \label{eq:S_sub}
        \end{align}
        there exists at least one codeword $ \xvhat \in \Cc_n $ for which $ (\xv, \xvhat) \in \Tc^n(\Ptilde_{X\Xhat})$.
    \end{enumerate}
\end{lem}

Once this lemma is established, the intuition is that the encoder chooses some codeword inducing a joint type $\Ptilde_{X\Xhat}$ that minimizes the $d_0$-distortion $\EE_{\Ptilde}[ d_0(X,\Xhat) ]$ subject to the above constraints, and these constraints asymptotically simplify to $\Ptilde_X = \Pi_X$,  $\Ptilde_{\Xhat} = Q_{\Xhat}$, and $I_{\Ptilde}(X;\Xhat) \le R$ as $n \to \infty$ and $\delta \to 0$.  Then, the $d_1$-distortion is upper bounded by the maximum $d_1$-distortion among all such minimizers, leading to the following.

\begin{lem} {\em (Achievability via the constant-composition ensemble \cite{Lap97})} \label{lem:existing}
    Under the mismatched rate-distortion setup with a discrete memoryless source $\Pi_X$ and distortion measures $(d_0,d_1)$, the following distortion is achievable at rate $R$ via constant-composition random coding with an auxiliary distribution $Q_{\Xhat} \in \Pc(\Xchat)$:
    \begin{gather}
        \Dibar(Q_{\Xhat},R) = \max_{\Ptilde_{X\Xhat} \in \Pctilde} \EE_{\Ptilde}[ d_1(X,\Xhat) ], \label{eq:D1bar}
    \end{gather}
    where
    \begin{align}\label{eq:setF}
        \Pctilde = \Bigg\{ \Ptilde_{X\Xhat} \,:\, \Ptilde_{X\Xhat} \in \argmin_{ \substack{ \Ptilde_{X\Xhat} \,:\, \Ptilde_X = \Pi_X, \Ptilde_{\Xhat} = Q_{\Xhat}, \\ I_{\Ptilde}(X;\Xhat) \le R } } \EE_{\Ptilde}[d_0(X,\Xhat)]\Bigg\}. 
    \end{align}
    Consequently, we have $D_1^*(R) \le \min_{Q_{\Xhat}} \Dibar(Q_{\Xhat},R)$.
\end{lem}

Formalizing the above intuition requires a rather technical continuity argument to transfer from joint types to general joint distributions and replace $\delta$ by zero.  Since these details were omitted in \cite{Sca20} and are slightly different from \cite{Lap97} due to the different ensemble used (constant-composition vs.~nearly-constant composition), we provide an overview in Appendix \ref{app:technical}, along with a discussion of how to adapt it to the other ensembles we consider.

\subsection{i.i.d.~Ensemble} \label{sec:iid}

While Lapidoth's results are based on the constant-composition ensemble (or very closely related ensembles), an even more common approach in general information theory problems is the {\em i.i.d.~random coding ensemble}, in which each codeword is independently drawn from
\begin{equation}
    P_{\Xvhat}(\xvhat) = \prod_{i=1}^n Q_{\Xhat}(\xhat_i)
\end{equation}
for some auxiliary distribution $Q_{\Xhat}$.  Note that unlike constant-composition coding, this ensemble naturally applies to continuous (or more general) alphabets.  Accordingly, it will be our main focus when considering general alphabets in Section \ref{sec:dual}.

For discrete memoryless sources, the following result follows via an almost identical argument to that of Lemma \ref{lem:existing} (see \cite{Lap97} and \cite[Ch.~4]{Sca20}), so the details are omitted. 

\begin{lem} {\em (Achievability via the i.i.d.~ensemble)} \label{lem:iid}
    Under the mismatched rate-distortion setup with a discrete memoryless source $\Pi_X$ and distortion measures $(d_0,d_1)$, the following distortion is achievable at rate $R$ via i.i.d.~random coding with an auxiliary distribution $Q_{\Xhat} \in \Pc(\Xchat)$:
    \begin{gather}
        \Dibar^{\rm iid}(Q_{\Xhat},R) = \max_{\Ptilde_{X\Xhat} \in \Pctilde^{\rm iid}} \EE_{\Ptilde}[ d_1(X,\Xhat) ], \label{eq:D1bar_IID}
    \end{gather}
    where
    \begin{align}\label{eq:setF_IID}
        \Pctilde^{\rm iid} = \Bigg\{ \Ptilde_{X\Xhat} \,:\, \Ptilde_{X\Xhat} \in \argmin_{ \substack{ \Ptilde_{X\Xhat} \,:\, \Ptilde_X = \Pi_X, \\ D(\Ptilde_{X\Xhat} \| \Pi_X \times Q_{\Xhat}) \le R } } \EE_{\Ptilde}[d_0(X,\Xhat)]\Bigg\}.
    \end{align}
    Consequently, we have $D_1^*(R) \le \min_{Q_{\Xhat}} \Dibar^{\rm iid}(Q_{\Xhat},R)$.
\end{lem}

Intuitively, the constraint $\Ptilde_{\Xhat} = Q_{\Xhat}$ is absent in \eqref{eq:setF_IID} because the codewords no longer necessarily have an empirical distribution equal to (or close to) $Q_{\Xhat}$.  In addition, the mutual information $I_{\Ptilde}(X;\Xhat)$ is replaced by $D(\Ptilde_{X\Xhat} \| \Pi_X \times Q_{\Xhat})$, since the two are no longer equivalent when $\Ptilde_{\Xhat}$ may differ from $Q_{\Xhat}$.


As we will see via a more general discussion in Section \ref{sec:structure}, neither the i.i.d.~nor constant-composition ensembles are guaranteed to outperform one another in general.

\vspace*{-1ex}
\section{Multi-User Random Coding Techniques} \label{sec:multi_user}
\vspace*{-0.5ex}

In this section, we present two achievability results based on multi-user coding techniques that were previously used for mismatched channel coding \cite{Lap96,Sca16a,Som15}, namely, superposition coding and parallel coding.  We also briefly mention that connections and differences between these ensembles (and other closely related ensembles) have been explored in detail in matched multi-user problems, e.g., see \cite{Cho08,Nai09,Wan13} and the references therein.

\subsection{Superposition Coding} \label{sec:SupCoding}
 
\subsubsection{Codebook construction}
The ensemble is defined in terms of an auxiliary alphabet $ \Uc $, an auxiliary codeword distribution $ P_{\Uv} \in \Pc(\Uc^n) $ and a conditional codeword distribution $ \Pc_{\Xvhat|\Uv} \in \Pc(\hat{\Xc}^n | \Uc^n) $. We fix $ (R_0, R_1) $ and generate the codebook in the following manner:
\begin{itemize}
    \item First, $ M_0 = \lfloor e^{nR_0} \rfloor $ codewords are drawn uniformly and independently from the $n$-letter distribution $ P_{\Uv} $, generating $ \{\Uv^{(i)}\}_{1\leq i \leq M_0} $.
    \item For each $ \uv^{(i)} \in \{\Uv^{(i)}\}_{1\leq i \leq M_0} $, $ M_1 = \lfloor e^{nR_1}\rfloor $ codewords are drawn uniformly and conditionally independently from the $n$-letter conditional distribution $P_{\Xvhat|\Uv}(\cdot | \uv^{(i)})$, generating $\{\Xvhat^{(i, j)}\}_{\begin{subarray}{1}
        1\leq i \leq M_0\\
        1\leq j \leq M_1
        \end{subarray}}$.
\end{itemize} 
For a given (joint) input distribution $Q_{U\Xhat} \in \Pc(\Uc \times \Xchat)$, let $ Q_{U\Xhat, n} \in \Pc_n(\Uc \times \hat{\Xc}) $ be a joint type such that $\|Q_{U\Xhat, n} - Q_{U\Xhat}\|_{\infty} \leq \frac{1}{n}$. We then consider the following (constant-composition) choices of $ P_{\Uv} $ and $ P_{\Xvhat|\Uv} $:
\begin{align}
    P_{\Uv}(\uv) &= \frac{1}{|\Tc^n(Q_{U, n})|} \openone\{\uv\in \Tc^n(Q_{U, n})\} \label{eq:supCodingU}\\
    P_{\Xvhat|\Uv}(\xvhat|\uv) &= \frac{1}{|\Tc_{\uv}^n(Q_{U\Xhat, n})|}\openone\{(\uv, \xvhat)\in \Tc^n(Q_{U\Xhat, n})\}. \label{eq:supcodingXhat}
\end{align}

The encoder maps a sequence $ \xv \in \Xc^n $ to $ (m_0, m_1) $ such that $ (m_0, m_1) = \argmin_{(i, j)} d_0^n(\xv,, \Xvhat^{(i, j)}) $, and the corresponding reconstruction sequence $ \xvhat $ is given by $ \xvhat = \Xvhat^{(m_0, m_1)} $. The rate is given by 
\begin{align}
    R = \frac{1}{n}\log(M_0 M_1) = R_0 + R_1.
\end{align}

Under this ensemble, we have the following achievability result, which we compare to Lemma \ref{lem:existing} in Sections \ref{sec:structure} and \ref{sec:examples}.

\begin{thm} \label{thm:supcoding_achievability}
{\em (Achievability for Superposition Coding)}
Under the superposition coding ensemble described above with input distribution $Q_{U\Xhat} \in \mathcal{P}(U \times \Xchat)$ and rates $R_0$ and $R_1$, the following distortion is achievable at rate $R = R_0 + R_1$:
\begin{align}
    \Dbar_1(Q_{U\Xhat}, R_0, R_1) &= \max_{\Ptilde_{XU\Xhat} \in \Pctilde} \mathbb{E}_{\Ptilde}[d_1(X, \Xhat)], \label{eq:d1_mm}
\end{align}
where
\begin{align} 
    \Pctilde &= \left\{\Ptilde_{XU\Xhat} : \Ptilde_{XU\Xhat} \in \argmin_{\substack{\Ptilde_{XU\Xhat} \,:\, \Ptilde_X = \Pi_X,\\
            \Ptilde_{U\Xhat} = Q_{U\Xhat},\\
            I_{\Ptilde}(X;U) \leq R_0, \\
            I_{\Ptilde}(X; U, \Xhat) \leq R_0 + R_1}} \mathbb{E}_{\Ptilde}[d_0(X, \Xhat)]\right\}. \label{eq:Ptilde_SC}
\end{align}
Consequently, we have $D_1^*(R) \leq \min_{\substack{(Q_{U\Xhat},R_1, R_2)\,:\, \\ R_1 + R_2 = R}} \Dbar_1(Q_{U\Xhat}, R_0, R_1) $.
\end{thm}

We provide a short proof outline here, and defer the full details to Appendix \ref{app:ach_proofs}.  
The analysis is based on the method of types.  The marginal constraints in \eqref{eq:Ptilde_SC} essentially follow immediately by construction.  The main effort is in showing that joint types with $I_{\Ptilde}(X;U) \leq R_0 - \delta$ and $I_{\Ptilde}(X; U, \Xhat) \leq R_0 + R_1 - \delta$ occur, but those with $I_{\Ptilde}(X;U) \ge R_0 + \delta$ or $I_{\Ptilde}(X; U, \Xhat) \ge R_0 + R_1 + \delta$ do not, in analogy with Lemma \ref{lem:occurence_of_types}.
    For both the existence and non-existence claims, we analyze the probability $P_{\rm existence} = \PP\big[ \bigcup_{i,j} \big\{ (\xv,\Uv^{(i)},\Xvhat^{(i,j)}) \in \Tc^n(\Ptilde_{XU\Xhat}) \big\} \big]$ for fixed $\Ptilde_{XU\Xhat}$, seeking to show that the probability rapidly approaches zero (non-existence) or one (existence).  We can upper bound $P_{\rm existence}$ using the truncated union bound $\PP\big[\bigcup_{i} A_i \big] \le \min\{1,\sum_{i} \PP[A_i]\}$ separately for the sums over $i$ and $j$, and we can get an essentially matching lower bound using the independence properties in the codebook construction.  Once that is done, the desired result is attained by first applying standard exponentially tight bounds on the relevant type class probabilities, and then using a continuity argument similar to that of Lapidoth \cite{Lap97}.

\subsection{Expurgated Parallel Coding}

We consider a codebook generated from two auxiliary codebooks of size $M_1$ and $M_2$ over auxiliary alphabets $ \hat{\Xc}_1 $ and $ \hat{\Xc}_2 $, along with a function $ \psi: \Xc_1\times\Xc_2 \rightarrow \hat{\Xc} $ mapping to the reconstruction alphabet.  Specifically, following analogous ideas used in mismatched channel coding \cite{Lap96}, the codebook is generated as follows:
\begin{itemize}
    \item For $ \nu = 1, 2 $, auxiliary constant-composition codebooks $\{ \Xvhat^{(i)}_{\nu} \}_{i=1}^{M_{\nu}}$ with input distributions $ Q_{X_{\nu}} $ and with $ M_\nu$ codewords each are generated, each independently drawn uniformly over the type class $\Tc^n(Q_{X_{\nu}, n})$ with $ \|Q_{\nu, n} - Q_\nu\|_\infty \leq \frac{1}{n} $.
    \item An initial codebook of size $M_1M_2$ is constructed as $ \Xvhat^{(i, j)} = \psi^n(\Xvhat^{(i)}_1, \Xvhat^{(j)}_2)$ for $1 \leq i \leq M_1, 1\leq j \leq M_2 $, where $\psi^n(\cdot,\cdot)$ applies the function $\psi$ entry-by-entry.
    \item The codebook is then expurgated, keeping only the codewords $\Xvhat^{(i, j)}$ whose corresponding pairs $ (\Xvhat^{(i)}_1, \Xvhat^{(j)}_2) $ have empirical distributions within $\delta > 0$ of $Q_{\Xhat_{1}}\times Q_{\Xhat_{2}}$ in $\ell_{\infty}$-norm.  Let $\Ic$ denote the set of $(i,j)$ pairs that are kept.
\end{itemize}

An input sequence $ \xv $ is mapped to $ (m_1, m_2) $ such that 
\begin{align}
    (m_1, m_2) &= \argmin_{(i,j) \in \Ic} d_0^n\big(\xv, \psi^n(\Xvhat^{(i)}_1, \Xvhat^{(j)}_2)\big).
\end{align}
The reconstruction sequence is given as $ \Xvhat = \psi^n(\Xvhat_1^{(m_1)}, \Xvhat_2^{(m_2)}) $, and the rate is given by $R = \frac{1}{n} \log |\Ic|$.  While this technically means that the rate is random, it will turn out that $R$ approaches $R_1 + R_2$ with probability approaching one, where $R_{\nu} = \frac{1}{n} \log M_{\nu}$.

Under this ensemble, we have the following achievability result, which we compare to Lemma \ref{lem:existing} in Sections \ref{sec:structure} and \ref{sec:examples}.

\begin{thm} \label{thm:parcoding_achievability}
    {\em (Achievability for Expurgated Parallel Coding)}
    For given auxiliary distributions $ Q_{X_{1}}, Q_{X_{2}} $, function $ \psi(\cdotp) $ and rates $ R_1, R_2 $, using expurgated parallel coding, the following distortion is achievable at rate $R = R_1 + R_2$:
    \begin{align}
    \Dbar_1(Q_{X_1}, Q_{X_2}, \psi, R_1, R_2) = \max_{\Ptilde_{X\Xhat_1, \Xhat_2} \in \Pctilde} \mathbb{E}_{\Ptilde}[d_1(X, \psi(\Xhat_1, \Xhat_2))],
    \end{align}
    where
    \begin{align} 
    \Pctilde &= \left\{\Ptilde_{X\Xhat_1\Xhat_2} : \Ptilde_{X\Xhat_1\Xhat_2} \in \argmin_{\substack{\Ptilde_{X\Xhat_1\Xhat_2}: \Ptilde_X = \Pi_X,\\
            \Ptilde_{\Xhat_1\Xhat_2} = Q_{\Xhat_1}\times Q_{\Xhat_2}, \\
            I_{\Ptilde}(X;\Xhat_1) \leq R_1, \\
            I_{\Ptilde}(X;\Xhat_2) \leq R_2, \\
            I_{\Ptilde}(X; \Xhat_1, \Xhat_2) \leq R_1 + R_2}} \mathbb{E}_{\Ptilde}[d_0(X, \psi(\Xhat_1, \Xhat_2))]\right\}. \label{eq:Ptilde_EX}
    \end{align}
    Consequently, we have $D_1^*(R) \leq \min_{\substack{(Q_{\Xhat_1}, Q_{\Xhat_2}, \psi, R_1, R_2)\,:\, \\ R_1 + R_2 = R}} \Dbar_1(Q_{\Xhat_1}, Q_{\Xhat_2}, \psi, R_1, R_2)$.
\end{thm}

    Once again, we only provide a short proof outline here, and defer the full details to Appendix \ref{app:ach_proofs}. 
    The main steps are analogous to those of Theorem \ref{thm:supcoding_achievability}, but in this case, the independence structure is different, since the codewords $\Xvhat^{(i, j)}$ and $\Xvhat^{(i', j')}$ are dependent whenever $i = i'$ or $j = j'$.  We can still use the truncated union bound to upper bound the relevant existence probabilities, whereas the lower bound instead uses de Caen's inequality \cite{Dec97}.  This is done in the same way as \cite{Sca16a}, and while the precise bound used in \cite{Sca16a} is slightly too loose due to a factor of $\frac{1}{4}$, we can easily sharpen their analysis by avoiding a step of the form $a+b+c+d \le 4 \max\{a,b,c,d\}$.  By doing so and again using standard properties of types followed by a continuity argument, we obtain the desired result.

\subsection{Structured vs.~Unstructured Random Codebooks} \label{sec:structure}

In the channel coding problem with mismatched decoding, it is known that (analogs of) the above ensembles never provide an achievable rate worse than that of independent constant-composition codewords with the same marginal codeword distribution.  It is therefore natural to ask whether the same holds for mismatched rate-distortion theory.

In fact, it is straightforward to see that for fixed choices of parameters $(Q_{U\Xhat},R_0,R_1)$ or $(Q_{\Xhat_1},Q_{\Xhat_2},R_1,R_2,\psi)$, the answer is negative.  To see this, fix the encoding function $d_0$, and consider any setup in which \eqref{eq:setF} has a unique minimizer $\Ptilde_{X\Xhat}^*$, and \eqref{eq:Ptilde_SC} or \eqref{eq:Ptilde_EX} has a unique minimizer with a joint marginal $\Ptilde_{X\Xhat}^{**}$ that differs from $\Ptilde_{X\Xhat}^*$.  The existence of examples satisfying these requirements can be inferred from \cite{Lap97} or from Section \ref{sec:examples}  to follow.

Then, consider the following possibilities for the true distortion function:
\begin{align}
    d_1^*(x,\xhat) = \log \Ptilde_{X\Xhat}^*(x,\xhat), \quad d_1^{**}(x,\xhat) = \log\Ptilde_{X\Xhat}^{**}(x,\xhat).
\end{align}
The non-negativity of KL divergence implies that any maximization problem of the form $\max_{P} \EE_{P}\big[\log Q(X)]$ is uniquely maximized by $P=Q$, and hence, we conclude that $\Ptilde_{X\Xhat}^*$ yields strictly higher distortion under $d_1^*$, whereas $\Ptilde_{X\Xhat}^{**}$ yields strictly higher distortion under $d_1^{**}$.  Hence, there is no general inequality between the achievable mismatched distortion functions.

On the other hand, when the parameters $(Q_{U\Xhat},R_0,R_1)$ or $(Q_{\Xhat_1},Q_{\Xhat_2},R_1,R_2,\psi)$ are {\em optimized}, we can easily guarantee being at least as good as the existing bound in Lemma \ref{lem:existing}.  This is because the superposition coding ensemble becomes equivalent to having independent codewords when $R_0 = 0$ (or $R_1 = 0$), and similarly for expurgated parallel coding with $R_2 = 0$ and $\psi(x_1,x_2) = x_1$ (or $R_1 = 0$ and $\psi(x_1,x_2) = x_2$).

More importantly, the benefit of multi-user coding techniques is seen by constructing examples where they provably meet the matched rate distortion function but the existing bound fails to do so.  In the next subsection, we revisit one such example from \cite{Lap97} for expurgated parallel coding, and provide a new example of this kind for superposition coding.

We note that the preceding discussion not only applies to independent codewords vs.~multi-user coding techniques, but also to i.i.d.~vs.~constant-composition random coding.  For fixed $Q_{\Xhat}$ and general choices of $d_0$ and $d_1$, it was already noted in \cite{Sca20} that neither ensemble is guaranteed to outperform the other, and the above arguments provide another way of seeing this.  On the other hand, when $d_0 = d_1$, the i.i.d.~ensemble is {\em always} as good or better, because the minimum in \eqref{eq:setF_IID} is less constrained compared to \eqref{eq:setF}, meaning there are more joint distributions to choose from for making $\EE_{\Ptilde}[d_0(X,\Xhat)]$ small.  Intuitively, the added diversity in the i.i.d.~random codebook leads to better (or equal) performance when the correct distortion measure is used for encoding, whereas when $d_0 \ne d_1$, the situation is more subtle depending on how the encoder may be led astray by the mismatch.  In Appendix \ref{app:example_tie}, we give an example with $d_0 \ne d_1$ where the constant-composition ensemble is strictly better.

\subsection{Examples} \label{sec:examples}

In this subsection, we present examples that illustrate improvements in the distortion-rate function when using the codebook generation techniques described in the previous section.

\subsubsection{Parallel Binary Source}

Consider a source with $ \Xc = \{0, 1\}^2 $ as both the source and reconstruction alphabet, with $ \Pi_X(x) = \frac{1}{4} $ for each $ x $. As in \cite{Lap97}, we consider the encoding metric 
\begin{align}
    d_0(x, \xhat)=\lambda\openone\{x_1 \neq \xhat_1\} + (1-\lambda)\openone\{x_2 \neq \xhat_2\},
\end{align}
and a distortion metric of 
\begin{align} 
    d_1(x, \xhat) = \frac{1}{2}(\openone\{x_1 \neq \xhat_1\} + \openone\{x_2 \neq \xhat_2\}).
\end{align}
For convenience we switch to working in bits (rather than nats) in this example.  It was shown in \cite{Lap97} that the achievable distortion with independent codewords (Lemma \ref{lem:existing}) simplifies to
\begin{gather}
    \Dbar_1(Q_{\Xhat}, R) = \frac{1}{2}(\tilde{\delta}_1^*+\tilde{\delta}_2^*) \label{eq:parallal_known1}\\
    (\tilde{\delta}_1^*, \tilde{\delta}_2^*) = \argmin_{\substack{(\tilde{\delta}_1, \tilde{\delta}_2) \,:\,  (1-H_2(\tilde{\delta}_1)) +(1-H_2(\tilde{\delta}_2)) \leq R}} \lambda \tilde{\delta}_1 + (1-\lambda)\tilde{\delta}_2, \label{eq:parallal_known2}
\end{gather}
whereas the matched distortion-rate function is given by $\Dbar_1^*$ such that 
\begin{equation}
    2(1-H_2(D_1^*)) = R. \label{eq:parallel_matched}
\end{equation}
The fact that parallel coding outperforms this result and attains the matched performance was already noted in \cite{Lap97} using direct arguments, namely, the fact that the two encoding metrics are equivalent when the codebook has a product structure.  Here we outline how the same (with expurgation) can be established via the analytical expression in Theorem \ref{thm:parcoding_achievability}, thus serving as a sanity check for this general result.

Consider taking $\psi(x_1,x_2) = (x_1,x_2)$ and  $Q_{\Xhat_\nu} = \big(\frac{1}{2}, \frac{1}{2}\big)$ ($\nu = 1, 2$) in Theorem \ref{thm:parcoding_achievability}. 
%
As was noted in \cite{Lap97}, the independence of $X_1$ and $X_2$ is sufficient to lower bound $I_{\Ptilde}(X; \Xhat_1, \Xhat_2)$ by $I_{\Ptilde}(X_1; \Xhat_1) + I_{\Ptilde}(X_2; \Xhat_2)$.  In addition, we trivially have $I_{\Ptilde}(X; \Xhat_{\nu}) \ge I_{\Ptilde}(X_{\nu}; \Xhat_{\nu})$ for $\nu = 1,2$.  Since both $X$ and $\Xhat$ follow the uniform distribution on $\{0,1\}^2$, the relevant marginals $\Ptilde_{\Xhat_{\nu}|X_{\nu}}$ must follow a binary symmetric channel law.  Denoting the associated crossover probabilities by $\tilde{\delta}_1$ and $\tilde{\delta}_2$, it follows that
\begin{gather}
    I_{\Ptilde}(X; \Xhat_1, \Xhat_2) \geq (1-H_2(\tilde{\delta}_1)) + (1-H_2(\tilde{\delta}_2)), \label{eq:mi_bound1} \\
    I_{\Ptilde}(X; \Xhat_{\nu}) \ge 1-H_2(\tilde{\delta}_{\nu}), \quad \nu = 1, 2. \label{eq:mi_bound2}
\end{gather}
Again following \cite{Lap97}, we observe that these bounds hold with equality when the underlying BSCs are independent, and that such independence must hold under the optimal $\Ptilde_{X\Xhat_1\Xhat_2}$.  This is because if any dependence were present, moving to the independent version would allow us to decrease $R_1$ and/or $R_2$ without increasing the distortion.

In view of \eqref{eq:mi_bound1}--\eqref{eq:mi_bound2} holding with equality, we are now left with the following analog of \eqref{eq:parallal_known1}--\eqref{eq:parallal_known2}:
\begin{gather}
    \Dbar_1(Q_{X_1}, Q_{X_2}, \psi, R_1, R_2) = \frac{1}{2}(\tilde{\delta}_1^*+\tilde{\delta}_2^*) \label{eq:parallel1}\\
    (\tilde{\delta}_1^*, \tilde{\delta}_2^*) = \argmin_{\substack{(\tilde{\delta}_1, \tilde{\delta}_2) \,:\, 1-H_2(\tilde{\delta}_1) \le R_1, 1-H_2(\tilde{\delta}_2) \le R_2}} \lambda \tilde{\delta}_1 + (1-\lambda)\tilde{\delta}_2, \label{eq:parallel2}
\end{gather}
where the constraint $(1-H_2(\tilde{\delta}_1)) +(1-H_2(\tilde{\delta}_2)) \leq R_1 + R_2$ is omitted because it is automatically guaranteed by the other two.  Finally, we set $R_1 = R_2 = \frac{R}{2}$, and observe that the minimum in \eqref{eq:parallel2} amounts to separate minimizations over the two parameters, each giving an optimal value of $\tilde{\delta}_{\nu}$ that satisfies $1 - H_2(\tilde{\delta}_{\nu}) = R/2$.  In view of \eqref{eq:parallel_matched}, the resulting $d_1$-distortion coincides with that attained in the matched case, as desired.

\subsubsection{Uniform Ternary Source} \label{sec:ex_ternary}

Consider the uniform ternary source with $\Pi_X(x) = \frac{1}{3}$ for $x \in \{0, 1, 2\} $. We consider an encoding metric $ d_0(x, \xhat) = \openone\{x \neq \xhat\}$ and a true distortion measure $ d_1(x, \xhat) = \openone\{x \neq \xhat, x \neq 2\}$. Moreover, we consider the case that both $\Pi_X$ and $Q_{\Xhat}$ are uniform:
\begin{align}
    \Pi_X = Q_X &= \begin{bmatrix}
        \frac{1}{3} & \frac{1}{3} & \frac{1}{3}
    \end{bmatrix}.
\end{align}
We first outline the results for the matched setting, and for unstructured random coding with mismatch (Lemma \ref{lem:existing}), and then turn to our own achievability result (Theorem \ref{thm:supcoding_achievability}).

{\bf Matched distortion-rate function.}
In the matched setting with distortion $d_1$ only, the achievable distortion-rate function with a fixed codeword distribution $Q_{\Xhat}$ is given as follows (e.g., see \cite{Dem02}, or Lemma \ref{lem:existing} with $d_0=d_1$): 
\begin{align}
    \Dbar_1^*(Q_{\Xhat},R) &= \min_{\substack{\Ptilde_{X\Xhat} : I_{\Ptilde}(X; \Xhat) \leq R\\\Ptilde_{X} = \Pi_X, \Ptilde_{\Xhat} = Q_{\Xhat}}} \EE_{\Ptilde}[d_1(X, \Xhat)]. \label{matchedMin}
\end{align}

{\bf Standard constant-composition coding.}
The achievable mismatched distortion-rate function achieved using standard constant-composition coding (i.e., independent codewords) was characterized by Lapidoth \cite{Lap97}, who showed that that there is a unique $\Ptilde_{X\Xhat} \in \Pctilde$ minimizing $\EE[d_0(X, \Xhat)]$ in \eqref{eq:D1bar}, with $\Ptilde_{\Xhat|X}$ being a ternary symmetric channel with some cross-symbol transition probability $\delta^*$ such that $I_{\Ptilde}(X;\Xhat) = R$.  The resulting distortion-rate function is given by $\Dbar_1^*(R) = \frac{2\delta^*}{3}$, where $\delta^*$ implicitly depends on $R$.  It is shown in \cite[Fig.~4.2]{Sca20} that the resulting rate-distortion trade-off is strictly worse than the matched case (for uniform ternary $Q_{\Xhat}$).


{\bf Superposition coding.}
We take the auxiliary distribution $ Q_{U\Xhat} $ as
\begin{align} \label{supdist}
    Q_{U\Xhat} = \begin{bmatrix}
    \frac{1}{3} & \frac{1}{3} & 0 \\
    0 & 0 & \frac{1}{3}
\end{bmatrix},
\end{align}
with the rows indexed by $U$ and columns by $\Xhat$. We then have the following result for superposition coding.

\begin{lem} \label{lem:ternarySource}
    In the ternary source example described above, for any rate $R$, there exists $0 \leq R_0 \leq R$ such that the distortion-rate function for mismatched encoding using superposition coding, with $Q_{U\Xhat}$ as described in \eqref{supdist} and $R_1 = R - R_0$, equals the distortion-rate function of the matched case (where $d_1$ is used for encoding) with $Q_{\Xhat}$ being the $\Xhat$-marginal of \eqref{supdist}.
\end{lem}

The proof is given in Appendix \ref{app:ternary}, and we provide only a brief outline here. We first characterize the mutual information terms associated with the joint distribution optimizing the matched distortion-rate function (see \eqref{matchedMin}).  This joint distribution has an associated value of $I(X;U)$, and we take $R_0$ to be this value.  We then study analogous mutual information terms associated with the mismatched distortion-rate function \eqref{eq:d1_mm}, and combine them to deduce that the matched and mismatched solutions are in fact identical (particularly with the help of the constraint $I_{\Ptilde}(X;U) \leq R_0$ in \eqref{eq:Ptilde_SC}).

By Lemma \ref{lem:ternarySource}, at least for this fixed auxiliary distribution $Q_{\Xhat}$, we have established an example where standard constant-composition coding falls short of the matched performance, but superposition coding does not.  This example is significantly different from the parallel source example, in that the result does not appear to follow from any direct arguments, and $R_0$ needs to be carefully chosen to obtain the desired result.

\vspace*{-1ex}
\section{General Alphabets} \label{sec:dual}
\vspace*{-0.5ex}

In this section, we turn to general alphabets (in particular, possibly countably infinite or continuous), which have been extensively studied in the mismatched channel coding problem \cite{Kap93,Mer95,Gan00,Sca16a,ScarlettThesis}, but to our knowledge, not for the mismatched rate-distortion problem that we consider.  On the other hand, this direction has been considered for a distinct notion of mismatched rate-distortion in \cite{Dem02,Kon06}; specifically, while these works consider $d_0 = d_1$, they use ``mismatched codebooks'' in the sense of having a suboptimal choice of $Q_{\Xhat}$.  Our analysis will directly build on the results from these works.

For the purpose of capturing the examples of practical interest, it is useful to think of $\Pi_X$ and $Q_{\Xhat}$ as being mass functions for finite or countably infinite alphabets, and density functions for continuous alphabets.  Formally, however, since we build on the tools in \cite{Kon06}, the alphabets may be as general as those therein, where the only restriction is that the source and reconstruction alphabets are Polish spaces such that all singletons are measurable.  

\subsection{Alternative Expression for the i.i.d.~Ensemble}

Extending the constant-composition result from Lemma \ref{lem:existing} to general alphabets appears to be difficult (and Theorems \ref{thm:supcoding_achievability} and \ref{thm:parcoding_achievability} even more so).  One might consider using cost-constrained random coding \cite{Gan00,Sca16a}, but the difficulty is that compared to channel coding, the techniques needed for the achievability result and ensemble tightness are reversed.  The ensemble tightness of cost-constrained random coding for channel coding currently remains open, and accordingly, establishing an achievability result in the setting that we consider is similarly challenging.

In view of this difficulty, in the remainder of this section, we consider the simpler {\em i.i.d.~random coding ensemble} (see Section \ref{sec:iid}), for which we present a useful reformulation of Lemma \ref{lem:iid}.  Here we make use of the matched $d_0$-distortion function with fixed $Q_{\Xhat}$, defined as
\begin{align}
    \Dobar^{\rm iid}(Q_{\Xhat},R) &= \min_{\substack{\Ptilde_{X\Xhat} \,:\, \Ptilde_{X} = \Pi_X \\ D(\Ptilde_{X\Xhat} \| \Pi_X \times Q_{\Xhat}) \leq R}} \EE_{\Ptilde}[d_0(X, \Xhat)]. \label{matchedMin0}
\end{align}
Moreover, as in \cite{Dem02,Kon06}, we restrict our attention to rates (and their associated distortion levels) in a restricted range.  Specifically, we define the following extreme values of the $d_0$-distortion:\footnote{In the finite-alphabet case, we can replace ${\rm ess inf}_{\Xhat \sim Q}$ by a minimum over the support of $Q_{\Xhat}$.}
\begin{gather}
    \domin = \EE_{\Pi}[ {\rm ess inf}_{\Xhat \sim Q} d_0(X,\Xhat) ], \quad \doprod = \EE_{\Pi \times Q}[d_0(X,\Xhat)]. \label{eq:D0_limits}
\end{gather}
Intuitively, $\domin$ is the average $d_0$-distortion that would be attained by an infinite-length codebook, and $\doprod$ is the average distortion that would be attained by a codebook with just a single random codeword.  Accordingly, we consider rates in the interval $R \in (0,R_{\max})$, where
\begin{equation}
    R_{\max} = \lim_{D_0 \to \domin} \Rbar^{\rm iid}(Q_{\Xhat}, D_0) \label{eq:R_max}
\end{equation}
with the limit taken from above, and where $\Rbar^{\rm iid}$ is the matched rate-distortion function defined as \cite{Dem02}
\begin{align}
    \Rbar^{\rm iid}(Q_{\Xhat}, D_0) 
    &= \min_{\substack{\Ptilde_{X\Xhat} \,:\, \Ptilde_X = \Pi_X, \\ \EE_{\Ptilde}[d_0(X,\Xhat)] \le D_0}} D(\Ptilde_{X\Xhat} \| \Pi_X \times Q_{\Xhat}). \label{eq:R_primal} 
\end{align}
Working in the range $R \in (0,R_{\max})$ (corresponding to $d_0$-distortion in $(\domin,\doprod)$) turns out to give a unique minimizer in \eqref{eq:setF_IID}, whereas for higher rates this is not necessarily the case (see Appendix \ref{app:example_tie} for an example).

Although Lemma \ref{lem:iid} is stated for discrete memoryless sources, the expressions in \eqref{eq:D1bar_IID}--\eqref{eq:setF_IID} can still be taken as valid definitions for general memoryless sources.  We proceed by stating an equivalent form for $\Dibar^{\rm iid}$ \eqref{eq:D1bar_IID}, and then turn to proving a counterpart to Lemma \ref{lem:iid} for general alphabets.

\begin{lem} \label{lem:dual_IID}
    Consider the mismatched rate-distortion problem with a given source $\Pi_X$ and auxiliary distribution $Q_{\Xhat}$, and suppose that $\doprod < \infty$ and $R \in (0,R_{\max})$ with $R_{\max}$ given in \eqref{eq:R_max}.  Then, the quantity $\Dibar^{\rm iid}(Q_{\Xhat},R)$ in \eqref{eq:D1bar_IID} can equivalently be expressed as 
    \begin{align}
        \Dibar^{\rm iid}(Q_{\Xhat},R) = \EE_{\Ptilde^*}[d_1(X, \Xhat)], \label{eq:D_iid_dual}
    \end{align}
    where $\Ptilde^*_{X\Xhat}$ is defined according to the Radon-Nikodym derivative\footnote{For mass functions or density functions, we can simplify this expression to $\Ptilde^*_{X\Xhat}(x,\xhat) =  \frac{ \Pi_X(x) Q_{\Xhat}(\xhat) e^{\lambda^* d_0(x,\xhat)} }{ \EE_{Q}[e^{\lambda^* d_0(x,\Xhat)}] }$.}
    \begin{equation}
        \frac{{\rm d}\Ptilde^*_{X\Xhat}}{{\rm d} (\Pi_X \times Q_{\Xhat})}(x,\xhat) =  \frac{ e^{\lambda^* d_0(x,\xhat)} }{ \EE_{Q}[e^{\lambda^* d_0(x,\Xhat)}] }, \label{eq:Pstar_IID}
    \end{equation}
    and where $\lambda^* \le 0$ is the unique value such that the function $\Lambda(\lambda) = \EE_{\Pi}\big[ \log \EE_{Q}[e^{\lambda d_0(X,\Xhat)}] \big]$ satisfies $\Lambda'(\lambda) = D_0$, with $D_0 = \Dobar^{\rm iid}(Q_{\Xhat},R)$.
\end{lem}

The proof is given in Appendix \ref{app:continuous}, and makes use of various findings from \cite{Dem02}, in which the matched setting is considered (i.e., $d_0 = d_1$, but we will use $d_0$ when discussing this setting), and the joint distribution $\Ptilde^*_{X\Xhat}$ is utilized.  The authors of \cite{Dem02} show that the rate-distortion function is given in \eqref{eq:R_primal}, and provide the following equivalent ``dual'' expression:
\begin{equation}
     \Rbar^{\rm iid}(Q_{\Xhat}, D_0) = \sup_{\lambda \le 0} \{ \lambda D_0 - \Lambda(\lambda) \}. \label{eq:R_dual}
\end{equation}
Thus, in our setting, the definition of $\lambda^*$ in Lemma \ref{lem:dual_IID} follows naturally from solving the maximum \eqref{eq:R_dual}; the existence and uniqueness of $\lambda^*$ (as per its definition) when $D_0 \in (\domin,\doprod)$ is established in \cite{Dem02}.  Once $\lambda^*$ is specified, the joint distribution in \eqref{eq:Pstar_IID} is also specified.  With the preceding findings and definitions in place, the proof of Lemma \ref{lem:dual_IID} amounts to showing that $\Ptilde^*_{X\Xhat}$ is the unique minimizer in \eqref{eq:setF_IID}.


\subsection{Achievability Result for General Alphabets}

Analogous to $\doprod$ in \eqref{eq:D0_limits}, we define
\begin{gather}
    \diprod = \EE_{\Pi \times Q}[d_1(X,\Xhat)], \label{eq:D1_limits}
\end{gather}
and we impose the assumption that $\diprod < \infty$, along with $\doprod < \infty$ (as hinted above, this is also assumed for the single distortion measure considered in \cite{Dem02,Kon06}).

%
%

Our main result for general alphabets is stated as follows.

\begin{thm} \label{thm:continuous}
    {\em (Achievable Distortion with General Alphabets)}
    Consider the general-alphabet mismatched rate-distortion setup with fixed $\Pi_X$ and $Q_{\Xhat}$ satisfying $\doprod < \infty$ and $\diprod < \infty$, let $R \in (0,R_{\max})$ with $R_{\max}$ defined in \eqref{eq:R_max}, and suppose that ties are broken in \eqref{eq:enc_rule} by choosing the tied codeword with the smallest index.  Then, the distortion $\Dibar^{\rm iid}(Q_{\Xhat},R)$ is achievable under i.i.d.~random coding, where $\Dibar^{\rm iid}$ can be equivalently be expressed in the form \eqref{eq:D1bar_IID}--\eqref{eq:setF_IID} or \eqref{eq:D_iid_dual}--\eqref{eq:Pstar_IID}.
\end{thm}

The proof is given in Appendix \ref{app:continuous}, and utilizes an almost-sure convergence result from \cite{Kon06} regarding the empirical distribution of $(\Xv,\Xvhat)$ when $\Xvhat$ is the first codeword in an infinite-length codebook to attain a pre-specified distortion level $D_0$.  As noted in \cite{Kon06}, their result generalizes and strengthens an earlier ``favorite type'' theorem for the finite-alphabet setting \cite{Zam01}.  

In our analysis, we slightly adapt the result of \cite{Kon06} to characterize the empirical $d_1$-distortion induced by $(\Xv,\Xvhat)$ (whereas \cite{Kon06} focuses on probabilities of events).  However, even after doing so, a notable difficulty is that in our setting $D_0$ is not pre-specified, as we consider the minimum-distortion encoder with a fixed rate $R$.  To address this, we consider an intersection of events (each corresponding the adapted result of \cite{Kon06} with a different $D_0$ value) over a rational subset of $D_0$ values, noting that a countable intersection of almost-sure events also holds almost surely.  From there, the proof consists of arguing that (i) with a fixed rate $R$, the $d_0$-distortion of the selected codeword approaches $\Dobar^{\rm iid}(Q_{\Xhat},R)$ in \eqref{matchedMin0}, and (ii) via the preceding almost-sure event and a somewhat technical continuity argument, the corresponding $d_1$-distortion converges to $\Dibar^{\rm iid}(Q_{\Xhat},R)$.


\subsection{Example: Gaussian Source with a Mismatched Distortion Function} \label{sec:cont_example}

As an example of Theorem \ref{thm:continuous}, we consider the case that $X \sim N(0,\sigma^2)$, $\Xhat \sim N(0,\tau^2)$, and $d_0(x,\xhat) = (x-\xhat)^2$, for some $\sigma^2,\tau^2 > 0$.  We will consider both the matched case (corresponding to $d_1 = d_0$), and the mismatched case with  $d_1(x,\xhat) = \openone\{ \sign(x) \ne \sign(\xhat) \}$ (i.e., the true distortion only seeks that the signs be recovered correctly).  This can be viewed a toy example of a scenario in which the compression is performed assuming the goal of accurate estimation, but it is then used only for the purpose of binary classification.  

In the matched case where $d_1$ is also used for encoding (i.e., $d_0$ plays no role), the problem becomes equivalent to compressing an equiprobable binary source, and the rate-distortion function is $R^*(D_1) = \log 2 - H_2(D_1)$ for $D_1 \in \big[0,\frac{1}{2}\big]$ \cite[Sec.~10.3.1]{Cov06}. 

Moreover, when $d_0$ is used for both encoding and for measuring the final distortion (i.e., $d_1$ plays no role), the rate-distortion trade-off for i.i.d.~random coding was characterized in \cite[Ex.~1]{Dem02}: We have $\domin = 0$, $\doprod = \sigma^2 + \tau^2$, and the following for $D_0 \in (\domin,\doprod)$ when $\Rbar^{\rm iid}$ is measured in nats:
\begin{gather}
    \Rbar^{\rm iid}(Q_{\Xhat}, D_0) =  \frac{1}{2}\log\frac{v}{D_0} - \frac{(v-D_0)(v-\sigma^2)}{2v\tau^2}, \label{eq:R_Gaussian}\\
    v = \frac{1}{2}\Big( \tau^2 + \sqrt{\tau^4 + 4D_0\sigma^2} \Big).
\end{gather}
Using this finding as a starting point, in the mismatched setting with the above choices of $d_0$ and $d_1$, we can evaluate the rate-distortion curve as follows:
\begin{itemize}
    \item Sweep over a range of values $\lambda \le 0$ and consider the corresponding joint distributions in $\Ptilde^*_{X\Xhat}$ in \eqref{eq:Pstar_IID} (implicitly depending on $\lambda$).
    \item Observe that due to the assumptions $X \sim N(0,\sigma^2)$, $\Xhat \sim N(0,\tau^2)$, and $d_0(x,\xhat) = (x-\xhat)^2$, all of the terms in \eqref{eq:Pstar_IID} give an expression whose exponent is quadratic with respect to $x$ and $\xhat$.  Thus, $\Ptilde^*_{X\Xhat}$ is a bivariate Gaussian distribution whose parameters $(\mu_X^*,\mu_{\Xhat}^*,\sigma_X^*.\sigma_{\Xhat}^*,\rho^*)$ (with $\rho^* \in [-1,1]$ being the correlation coefficient) can be computed as a function of $(\sigma,\tau,\lambda^*)$.  Omitting the tedious calculations, we state the resulting parameters as follows.  The coefficients to $-x^2$, $-\xhat^2$, and $-x\xhat$ in the exponent in \eqref{eq:Pstar_IID} are easily computed to be
    \begin{equation}
        a = \frac{1}{2\sigma^2} - \lambda, ~b = \frac{1}{2\tau^2} - \lambda, ~c = 2\lambda
    \end{equation}
    (with the remaining coefficients to $x$ and $\xhat$ being zero), and from these values, the desired parameters can be shown to  be as follows:
    \begin{gather}
        \mu_X^* = 0, ~\mu_{\Xhat}^* = 0, ~(\sigma_X^*)^2 = \sqrt{\frac{2b}{4ab-c^2}}, ~(\sigma_{\Xhat}^*)^2 = \sqrt{\frac{2a}{4ab-c^2}}, ~\rho^* = \frac{-c}{2\sqrt{ab}}. \label{eq:params}
    \end{gather}
    \item Given the parameters in \eqref{eq:params}, we can readily compute the $d_0$-distortion as $D_0 = (\sigma_X^*)^2 + (\sigma_{\Xhat}^*)^2 - 2\rho^* \sigma_X^* \sigma_{\Xhat}^*$ (by expanding the square in $\EE[(X-\Xhat)^2]$), and the $d_1$-distortion as the probability of the bivariate Gaussian lying in the second or forth quadrant of $\RR^2$ (e.g., using standard libraries for computing the multivariate CDF).
    \item With $D_0$ now known, the corresponding rate can be obtained by numerically inverting \eqref{eq:R_Gaussian}, e.g., using a binary search procedure.
\end{itemize}
The resulting rate-distortion curves are illustrated in Figure \ref{fig:continuous} (after converting from nats to bits) under the choices $\sigma^2 = \tau^2 = 1$.  Although nearest-neighbor encoding (i.e., using $d_0$) may sound reasonable for the purpose of preserving signs, it is seen to be highly suboptimal here.  This is particularly the case at low distortion levels, where the rate becomes unbounded in the mismatched case, despite a rate of one being trivial.  Intuitively, this is because typical $\Xv$ sequences contain many low-valued entries, and using the nearest-neighbor encoder is likely to flip their signs unless there are a huge number of codewords to choose from (i.e., a high rate).

\begin{figure}[!t]
    \centering
    \includegraphics[width=0.5\textwidth]{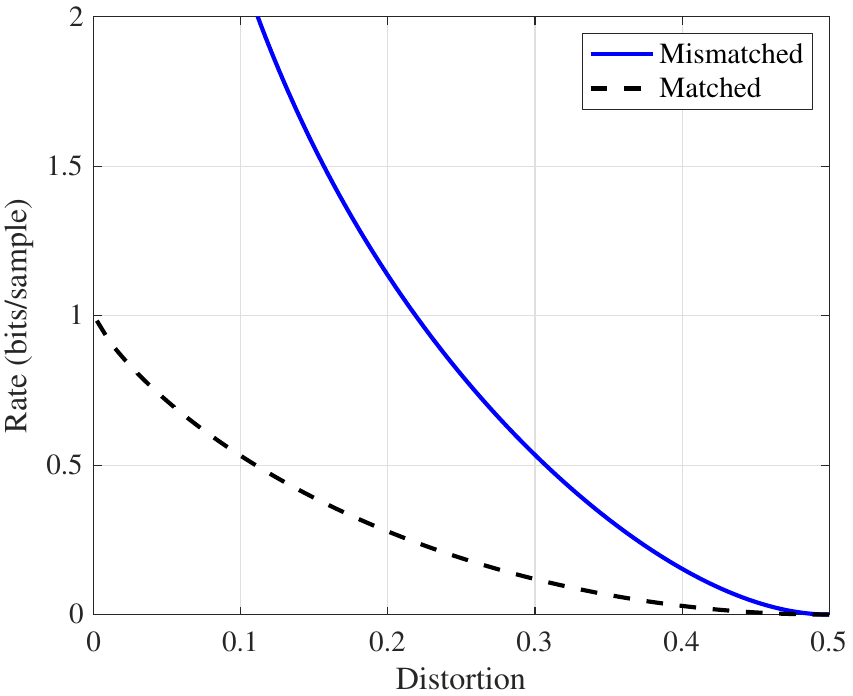}
    \par
    \caption{Rate-distortion curves for the Gaussian source with $\sigma^2 = \tau^2 = 1$; in the mismatched case, the distortion is 0-1 valued based on the signs. \label{fig:continuous}}
    \vspace*{-3ex}
\end{figure}

\vspace*{-1ex}
\section{Conclusion}
\vspace*{-0.5ex}

We have studied multi-user coding techniques and general alphabets for the mismatched rate-distortion problem, and established several new results that serve as natural counterparts to those that were previously known for mismatched channel coding.  Several possible directions remain for future work, including (i) studying the refined superposition coding ensemble considered for channel coding in \cite{Sca16a}; (ii) establishing conditions under which the various ensembles provably outperform one another; and (iii) attaining general-alphabet counterparts of the rate-distortion curves achieved by the constant constant-composition and/or multi-user ensembles.

\vspace*{-1ex}
\appendices

\vspace*{-1ex}
\section{Analysis of the Uniform Ternary Source Example} \label{app:ternary}
\vspace*{-0.5ex}

In this section, we study the ternary source example of Example \ref{sec:ex_ternary}, and prove Lemma \ref{lem:ternarySource} therein.  Throughout this appendix, all information measures are in units of nats.

\subsection{Alternative Expressions}

Since both the $X$-marginal and $\Xhat$-marginal of $\Ptilde_{X\Xhat}$ are uniform, we can express $\Ptilde_{X\Xhat}$ in the following form:
\begin{align}
    \Ptilde_{X\Xhat} &= \begin{bmatrix}
        \frac{1}{3} - p_{01} - p_{02} & p_{01} & p_{02}\\
        p_{10} & \frac{1}{3} - p_{10} - p_{12} & p_{12}\\
        p_{01} + p_{02} -p_{10} & p_{10} + p_{12} - p_{01} & \frac{1}{3} - p_{02} - p_{12}
    \end{bmatrix}, \label{genDist}
\end{align}
where we write the joint distribution in matrix form with rows indexed by $x$ and columns indexed by $\xhat$, and $(p_{01}, p_{10}, p_{02}, p_{12})$ are constrained to take values such that all elements of the matrix are non-negative.  In this notation, the matched distortion in \eqref{matchedMin} can be rewritten as
\begin{align}
    \Dbar_1^*(Q_{\Xhat},R) =  \min_{\Ptilde_{X\Xhat} : I_{\Ptilde}(X; \Xhat) \leq R} p_{01} + p_{10} + p_{02}  + p_{12}. \label{matchedMin2}
\end{align}
To make this more explicit, we write $I_{\Ptilde}(X;\Xhat)$ in terms of the ternary entropy function $H_3(a,b)$ (i.e., the entropy with probabilities $(a,b,1-a-b)$) as follows:
\begin{align}
    I_{\Ptilde}(X;\Xhat) &= H_{\Ptilde}(\Xhat) - H_{\Ptilde}(\Xhat|X) \label{initJensen}\\
    &= \log 3 - H_{\Ptilde}(\Xhat|X) \\
    &= \log 3 - \sum_x \Pi_X(x)H_{\Ptilde}(\Xhat | X = x)\\
    &= \log 3 - \frac{1}{3}\Big(H_3(3p_{01}, 3p_{02}) + H_{3}(3p_{10}, 3p_{12})  \notag\\
    & \qquad + H_3(3p_{02} + (3p_{01} - 3p_{10}), 3p_{12} - (3p_{01} - 3p_{10}))\Big)\label{genMutInf}\\
    &\geq \log 3 - \frac{1}{3}\left(2H_3\left(3\frac{p_{01} + p_{10}}{2}, 3\frac{p_{02} + p_{12}}{2}\right) + H_3\left(3\frac{p_{02} + p_{12}}{2}, 3\frac{p_{02} + p_{12}}{2}\right) \right), \label{mutInfJensen}
\end{align}
where \eqref{genMutInf} follows since each $H(\Xhat|X=x)$ is the ternary entropy associated with multiplying the relevant row of \eqref{genDist} by 3 (i.e., dividing by $\Pi_X(x) = \frac{1}{3}$), and \eqref{mutInfJensen} follows by applying Jensen's inequality to the first two terms, and using $H_3(a,b) \le H_3\big(\frac{a+b}{2},\frac{a+b}{2}\big)$ (again exploiting the concavity of entropy) in the third term. Moreover, since the ternary entropy function is \emph{strictly} concave, equality holds if and only if $ p_{01} = p_{10} $ and $ p_{02} = p_{12} $.

\subsection{Analysis of the Matched Performance}

As a step towards proving Lemma \ref{lem:ternarySource}, which states that the matched performance is attainable even in the presence of mismatch, it is useful to first study the minimization problem \eqref{matchedMin} associated with the matched setting.  We present several useful lemmas regarding this.

\begin{lem} \label{symmLemma}
    In the present ternary source example, for $R \in [0,\log 3]$, any distribution achieving the minimum in \eqref{matchedMin} satisfies both $ p_{01} = p_{10} $ and $ p_{02} = p_{12} $ when expressed as given in \eqref{genDist}. 
\end{lem}
\begin{proof}
    This result trivially holds for $R = \log 3$, since the constraint $I_{\Ptilde}(X;\Xhat) \le R$ is then satisfied by all valid distributions of the form \eqref{genDist}. Selecting $\Xhat = X$ (i.e., $p_{01} = p_{02} = p_{10} = p_{12} = 0$) satisfies this with zero distortion.  Conversely, if $D_1^* = 0$ then the choice of distortion function implies that $p_{01} = p_{02} = p_{10} = p_{12} = 0$ (see \eqref{genDist}), which in turn implies $\Xhat = X$ and $R = \log 3$.  Thus, whenever $R < \log 3$, it must be the case that $\Dbar_1^*(Q_X,R) > 0$, and we focus on this scenario in the rest of the analysis.
    
    Fix $R < \log 3$, and assume for contradiction that some distribution $ \Ptilde_a $ is asymmetric (in the sense that $ p_{01, a} \neq p_{10, a} $ and/or $ p_{02, a} \neq p_{12, a} $) but achieves the minimum in \eqref{matchedMin}.  We define the symmetrized version $\Ptilde_s = \frac{\Ptilde_a + \Ptilde_a'}{2}$, where $ \Ptilde_a' $ is the distribution with $ p_{01, a}' = p_{10, a} $, $ p_{10, a}' = p_{01, a} $, $ p_{02, a}' = p_{12, a} $, and $ p_{12, a}' = p_{02, a} $ (i.e., the suitable pairs are interchanged). Observe that $\Ptilde_s$ is symmetric in the sense that $ p_{01, s} = p_{10, s} $ and $ p_{02, s} = p_{12, s} $.
    
    We now have the following:
    \begin{equation}
        I_{\Ptilde_s}(X; \Xhat) < I_{\Ptilde_a}(X; \Xhat) = I_{\Ptilde_a'}(X; \Xhat),
    \end{equation}
    where the inequality follows from \eqref{mutInfJensen} (and the condition for strict inequality stated just after), and the equality follows since the mutual information is invariant to interchanging the pairs as described above (see \eqref{genMutInf}).  Moreover, since the distortion is given by $p_{01} + p_{10} + p_{02}  + p_{12}$ (see \eqref{matchedMin2}), we have $ \EE_{\Ptilde_a}[d_{1}(X;\Xhat)] = \EE_{\Ptilde_s}[d_{1}(X;\Xhat)] $. Thus, the fact that $ I_{\Ptilde_a}(X; \Xhat) \leq R $ implies that $ I_{\Ptilde_s}(X; \Xhat) < R $, which implies that $ \Ptilde_s $ also achieves the minimum in \eqref{matchedMin}.
    
    
    
    If $p_{01,s} > 0$, then from the continuity of mutual information, there exists some $p_{01, s}^{\dagger} < p_{01, s}$ such that the corresponding mutual information still satisfies $I_{\Ptilde_s^{\dagger}}(X; \Xhat) \leq R$ (with $p_{02, s}^{\dagger}  = p_{02,s}$). The corresponding value of $\EE_{\Ptilde_s^{\dagger}}[d_1(X;\Xhat)]$ is then smaller than that of $\Ptilde_a$, thus contradicting the assumption that $\Ptilde_a$ achieves the minimum.
    
    If $p_{01,s} = 0$, then we must have $p_{02,s}$, since we have assumed that we are not in the zero-distortion region.  Then, by the same reasoning as the case where $p_{01,s} > 0$, we find that there exists $p_{02, s}^{\dagger} < p_{02, s}$ such that $I_{\Ptilde^{\dagger}}(X; \Xhat) \leq R$. This leads to a distortion smaller than the minimizing distribution, arriving at a contradiction and completing the proof.    
    \end{proof}

\begin{lem} \label{symCorr}
    In the present ternary source example, for $R \in [0,\log 3]$, any distribution achieving the minimum in \eqref{matchedMin} satisfies $ I_{\Ptilde}(X;\Xhat) = R $.
\end{lem}
\begin{proof}
    Suppose for contradiction that there exists a distribution $\Ptilde_s$ (assumed symmetric according to Lemma \ref{symmLemma}) achieving the minimum with $ I_{\Ptilde_s}(X;\Xhat) <R $.  Then, from the analysis in the proof of Lemma \ref{symmLemma}, there exists a distribution $\Ptilde_s'$ such that $\EE_{\Ptilde_s'}[d_1(X, \Xhat)] < \EE_{\Ptilde_s}[d_1(X, \Xhat)]$ and $I_{\Ptilde}(X;\Xhat) \le R$. This is a contradiction, since we have assumed that $\Ptilde_s$ achieves the minimum.
\end{proof}

In view of Lemma \ref{symmLemma}, we can reduce the four parameters $(p_{01},p_{10},p_{02},p_{12})$ in \eqref{genDist} to just two parameters  $(p_{01},p_{02})$.  To satisfy the constraint of non-negative matrix entries, we require the following:
\begin{gather}
    0 \le p_{01} + p_{02} \leq \frac{1}{3} \label{supconstr1}\\
    0 \le p_{02} \leq \frac{1}{6}. \label{supconstr2}
\end{gather}
Moreover, the symmetry from Lemma \ref{symmLemma} implies that \eqref{mutInfJensen} holds with equality, and the mutual information becomes
\begin{equation}
    I_{\Ptilde}(X;\Xhat) = \log 3 - \frac{1}{3}\left(2H_3\left(3p_{01}, 3p_{02}\right) + H_3\left(3p_{02}, 3p_{02}\right) \right). \label{eq:mi_simplified}
\end{equation}
Hence, \eqref{matchedMin2} reduces to minimizing $2(p_{01}+p_{02})$ subject to \eqref{supconstr1}--\eqref{supconstr2}, as well as \eqref{eq:mi_simplified} being at most $R$.

From the two-parameter expression \eqref{eq:mi_simplified}, we can readily evaluate the partial derivatives with respect to $p_{01}$ and $p_{02}$ (the details amount to standard calculus and are omitted):
\begin{gather}
    \frac{\partial }{\partial p_{01}} I_{\Ptilde}(X;\Xhat) = 2\log\left(\frac{p_{01}}{\frac{1}{3} - p_{01} - p_{02}}\right), \label{infpartial_p01}\\
    \frac{\partial }{\partial p_{02}} I_{\Ptilde}(X;\Xhat) = 2\log\left(\frac{p_{02}^2}{(\frac{1}{3} - p_{01} - p_{02})(\frac{1}{3} - 2p_{02})}\right). \label{infPartial_p02}
\end{gather}
Our final useful lemma regarding the matched case is as follows.

\begin{lem}
    In the present ternary source example with $R \in [0,\log 3]$, the optimal value of $p_{02}$ for a distribution achieving the matched distortion-rate function, when $d_1$ is used as the distortion metric, is at most $\frac{1}{9}$. \label{p02IneqLemma}
\end{lem}
\begin{proof}
We first note that the constraint $p_{01} + p_{02} \le \frac{1}{3}$ in \eqref{supconstr1} is inactive, because the choice $p_{01}  = p_{02} = 1/9$ gives a smaller value of distortion (namely, $2p_{01} + 2p_{02}$) compared to any pair $(p_{01}, p_{02})$ such that $ p_{01} + p_{02} = 1/3 $, and is guaranteed to be feasible since it gives $I_{\Ptilde}(X, \Xhat) = 0$.

Moreover, the constraint $p_{02} \le \frac{1}{6}$ in \eqref{supconstr2} is inactive, since for any fixed value of $p_{01}$, \eqref{infPartial_p02} implies that $I_{\Ptilde}(X, \Xhat)$ increases to its value at $p_{02} = \frac{1}{6}$ infinitely sharply.  Hence, when $p_{02} = \frac{1}{6}$, we can hold $p_{01}$ constant and slightly decrease $p_{02}$, thereby decreasing both $I_{\Ptilde}(X;\Xhat)$ and the distortion. 


We also claim that in the positive-distortion regime (i.e., $p_{01} + p_{02} > 0$), neither of $p_{01}$ and $p_{02}$ can be $0$. This is seen by fixing $p_{01} + p_{02}$ to a positive value $p_{\rm sum}$ and observing from \eqref{infpartial_p01}--\eqref{infPartial_p02} that
\begin{align}
    \lim_{p_{0\nu}\rightarrow 0}\frac{\partial I_{\Ptilde}(X;\Xhat)}{\partial p_{0\nu}} \rightarrow -\infty, \quad \nu = 1, 2;
\end{align}
whereas the partial derivative at $p_{0\nu} = p_{\rm sum}$ ($\nu = 1,2$) is bounded and finite.
If only one of the two is zero (say $p_{01}$), we see that increasing $p_{01}$ and decreasing $p_{02}$ by the same small quantity leads to a net decrease in $I_{\Ptilde}(X;
\Xhat)$ while keeping the distortion constant, contradicting the fact that $I_{\Ptilde}(X;\Xhat) = R$ (from Lemma \ref{symCorr}).

Combining the preceding paragraphs, we have shown that the solution to the symmetrized version of \eqref{matchedMin2} is unchanged when the constraints in \eqref{supconstr1}--\eqref{supconstr2} are dropped altogether.  Hence, we are left with studying the simplified problem
\begin{equation}
    \min_{\substack{(p_{01},p_{02}) \,:\,  \log 3 - \frac{1}{3}\left(2H_3\left(3p_{01}, 3p_{02}\right) + H_3\left(3p_{02}, 3p_{02}\right) \right) \le R}}  2(p_{01} + p_{02}), \label{eq:matched_reduced}
\end{equation}
where the expression in the constraint comes from \eqref{eq:mi_simplified}.

Applying the KKT optimality conditions to \eqref{eq:matched_reduced}, we find that there exists $\lambda > 0$ such that
\begin{align}
    \nabla I_{\Ptilde^*}(X;\Xhat) & = -\frac{1}{\lambda}\begin{bmatrix}
        2 \\ 2
    \end{bmatrix},
\end{align}
from which it follows that
\begin{gather}
     \left.\frac{\partial}{\partial p_{01}}I_{\Ptilde^*}(X;\Xhat) \right|_{p_{01}^*} = \left.\frac{\partial}{\partial p_{02}}I_{\Ptilde^*}(X;\Xhat) \right|_{p_{02}^*} \label{matchedKKT}\\
    \implies p_{01}^* = \frac{(p_{02}^*)^2}{\frac{1}{3} - 2p_{02}^*}\label{rel_p01_p02},
\end{gather}
where \eqref{rel_p01_p02} follows by equating \eqref{infpartial_p01} with \eqref{infPartial_p02} and simplifying.

A simple analysis of the derivative of \eqref{rel_p01_p02} reveals that $p_{01}^*$ is a strictly increasing function of $p_{02}^*$ in the admissible range $p_{02}^* \in \big[0,\frac{1}{6}\big]$, and therefore, so is the distortion.  On the other hand, we know that $p_{01}^* = p_{02}^* = \frac{1}{9}$ corresponds to a rate of zero. By the non-increasing nature of the distortion-rate function, we conclude that distortion cannot exceed its zero-rate value, and therefore, $p_{02}$ (and in fact, also $p_{01}$) cannot exceed $\frac{1}{9}$.
\end{proof}

\subsection{Analysis of Superposition Coding}
Since the $X$-marginal and $\Xhat$-marginal are still uniform here, we can again consider $\Ptilde_{X\Xhat}$ in the form given in \eqref{genDist}.  Moreover, the following analog of Lemma \ref{symmLemma} holds. 

\begin{lem} \label{symmLemma2}
    Consider the minimization problem with respect to $\Ptilde_{XU\Xhat}$ in \eqref{eq:Ptilde_SC}.  In the present ternary source example, $\Ptilde_{XU\Xhat}$ is uniquely determined by $\Ptilde_{X\Xhat}$, and any distribution achieving the minimum satisfies both $p_{01} = p_{10}$ and $p_{02} = p_{12}$ when $\Ptilde_{X\Xhat}$ is expressed as given in \eqref{genDist}.
\end{lem}
\begin{proof}
    The choice of $Q_{U\Xhat}$ in \eqref{supdist} implies that $U$ is deterministic given $\Xhat$, which means that once $\Ptilde_{X\Xhat}$ is specified, so is $\Ptilde_{XU\Xhat}$ (recalling that we constrain $\Ptilde_{U\Xhat} = Q_{U\Xhat}$).  This proves the first claim.
    
    For the second claim, the proof is similar to that of Lemma \ref{symmLemma}, so we only outline the differences.  In the symmetrization step,  interchanging $p_{01} \leftrightarrow p_{10}$ and/or $p_{02} \leftrightarrow p_{12}$ still does not impact the distortion (this time $d_0$ instead of $d_1$).  Since $U$ is deterministic given $\Xhat$, the mutual information $I(X; U, \Xhat)$ reduces to $I(X;\Xhat)$, so can be handled similarly to the matched case.  Moreover, using the fact that $U = \openone\{ \Xhat = 2 \}$, we find that the marginal $\Ptilde_{UX}$ can be obtained by combining the first two columns of \eqref{genDist} to obtain
    \begin{align}
        \Ptilde_{XU} &= \begin{bmatrix}
            \frac{1}{3} - p_{02} & p_{02}\\
            \frac{1}{3} - p_{12} & p_{12}\\
            p_{02} + p_{12} & \frac{1}{3} - p_{02} - p_{12}
        \end{bmatrix}, \label{genDist_XU}
    \end{align}
    which has no dependence on $(p_{01},p_{10})$.  Then, following similar steps to \eqref{initJensen}--\eqref{mutInfJensen} (see also \eqref{uxInf} below), we obtain that $I_{\Ptilde}(X;U)$ is invariant to interchanging $p_{02} \leftrightarrow p_{12}$, and can always be made smaller by symmetrizing with respect to this interchanging (or it remains the same if such symmetry already held).
    
    With these findings in place, the symmetrization argument follows in the same way as the proof of Lemma \ref{symmLemma}, considering both mutual information terms when handling $(p_{02},p_{12})$, but only considering $I_{\Ptilde}(X;\Xhat)$ when handling $(p_{01},p_{10})$ (since $I_{\Ptilde}(X;U)$ does not depend on these values).
\end{proof}

Using Lemma \ref{symmLemma2} and recalling the choice of $Q_{U\Xhat}$ in \eqref{supdist}, Theorem \ref{thm:supcoding_achievability} simplifies as follows:
\begin{align}
    \Dbar_1^*(Q_{U\Xhat},R_0,R_1) = \max_{(p_{01},p_{02}) \,:\, \Ptilde_{XU\Xhat} \in \Pctilde} 2(p_{01} + 2p_{02}),
\end{align}
where 
\begin{align}
    \Pctilde = \left\{ \argmin_{\substack{(p_{01},p_{02})\,:\, I_{\Ptilde}(X; U) \leq R_0 \\ I_{\Ptilde}(X; U, \Xhat) \leq R}} 2(p_{01} + p_{02}) \right\}, \label{eq:Ptilde_ternary} 
\end{align}
and where the correspondence between $(p_{01},p_{02})$ and $\Ptilde_{XU\Xhat}$ is as follows:
\begin{equation}
     P_U(0) = \frac{2}{3}, P_U(1) = \frac{1}{3},
\end{equation}
\begin{align}
    \Ptilde_{X\Xhat | U = 0} &= \frac{3}{2}\begin{bmatrix}
        \frac{1}{3} - p_{01} - p_{02} & p_{01} & 0\\
        p_{01} & \frac{1}{3} - p_{01} - p_{02} & 0\\
        p_{02} & p_{02} & 0,
    \end{bmatrix}, \quad
    \Ptilde_{X\Xhat|U = 1} = 3\begin{bmatrix}
        0&0&p_{02}\\
        0&0&p_{02}\\
        0&0&\frac{1}{3} - 2p_{02}.
    \end{bmatrix}
\end{align}
Once again, $p_{01}$ and $p_{02}$ are constrained such that the matrix entries are non-negative.

With this symmetrization in place, the marginal distribution $\Ptilde_{XU}$ shown in \eqref{genDist_XU} simplifies to
\begin{align}
    P_{XU} &=\begin{bmatrix}
        \frac{1}{3} - p_{02} & p_{02}\\
        \frac{1}{3} - p_{02} & p_{02}\\
        2p_{02} & \frac{1}{3} - 2p_{02} \label{eq:PXU}
    \end{bmatrix},
\end{align}
which yields the following (with $H_2(\cdot)$ being the binary entropy function):
\begin{align}
    I_{\Ptilde}(X;U) &= H(U)  - H(U|X)\\
    &= \log 3 - \frac{2}{3}\log 2 -\sum_x H(U|X=x)\Pi_X(x)\\
    &= \log 3 - \frac{2}{3}\log 2 - \frac{1}{3}(2H_2(3p_{02}) + H_2(6p_{02})), \label{uxInf}
\end{align}
since each $H(U|X=x)$ is the binary entropy associated with multiplying the relevant row of \eqref{eq:PXU} by 3 (i.e., dividing by $\Pi_X(x) = \frac{1}{3}$).  Moreover, we have already established that $I_{\Ptilde}(X;U, \Xhat) = I_{\Ptilde}(X;\Xhat)$.
Using these results, we are now in a position to prove our main finding.

\begin{proof}[Proof of Lemma \ref{lem:ternarySource}]
    Recall that we simplified the $d_1$-distortion to $2(p_{01} + p_{02})$ using \eqref{genDist} with $p_{01} = p_{10}$ and $p_{02} = p_{12}$.  For the $d_0$-distortion, the first two entries in the final row also contribute, giving a value of $2(p_{01} + 2p_{02})$.  As with Lemma \ref{symmLemma}, since $X = \Xhat$ trivially satisfies all constraints when $R = \log 3$, we focus on the cases where $R < \log 3$, and consequentially, $p_{01} + p_{02} > 0$.  
    
    We claim that the simplified matched problem \eqref{eq:matched_reduced} is uniquely minimized by some values $p^*_{01, m}$ and $p^*_{02, m}$.  To see this, assume by contradiction that there are multiple minimizers.  The average of any two of these will have the same distortion as the minimum, while $I_{\Ptilde}(X;\Xhat) = \log 3 - \frac{1}{3}\left(2H_3\left(3p_{01}, 3p_{02}\right) + H_3\left(3p_{02}, 3p_{02}\right) \right)$ will be strictly less than $R$ (due to strict concavity), contradicting Lemma \ref{symCorr}.
    
    We note from from Lemma \ref{p02IneqLemma} that $p_{02,m}^* \leq \frac{1}{9}$, and we consider the following choice of the parameter $R_0 = R - R_1$ in superposition coding: 
    \begin{equation}
        R_0 = I_{\Ptilde_m^*}(X;U),
    \end{equation}
    where $I_{\Ptilde_m^*}(X;U)$ is the value obtained on substituting $ p_{02, m}^* $ in \eqref{uxInf}.
    For all $(p_{01}, p_{02})$ distinct from $(p_{01, m}^*, p_{02, m}^*)$ such that $I_{\Ptilde}(X; U) \leq R_0$ and $I_{\Ptilde}(X;\Xhat) \le R$, we have 
    \begin{equation}
        p_{01, m}^* + p_{02, m}^* < p_{01} + p_{02}, \label{eq:ineq1}
    \end{equation}
    since the minimizer of $d_1$ for the matched case is unique.  Moreover, we claim that 
    \begin{equation}	
        p_{02, m}^* \leq p_{02}. \label{eq:ineq2}
    \end{equation}
    This is seen by combining the fact that $p_{02,m}^* \leq \frac{1}{9}$ (established above) with the fact that $I_{\Ptilde}(X;U)$ is monotonically decreasing for $p_{02} \le \frac{1}{9}$ (see Figure \ref{fig:inf_ux_vs_p02}; an analytical proof is also straightforward).  Thus, under the constraint  $I_{\Ptilde}(X; U) \leq R_0$, it is not feasible to go below $p_{02,m}^*$.
    
    \begin{figure}[t!] 
        \centering
        \includegraphics[width=0.5\linewidth]{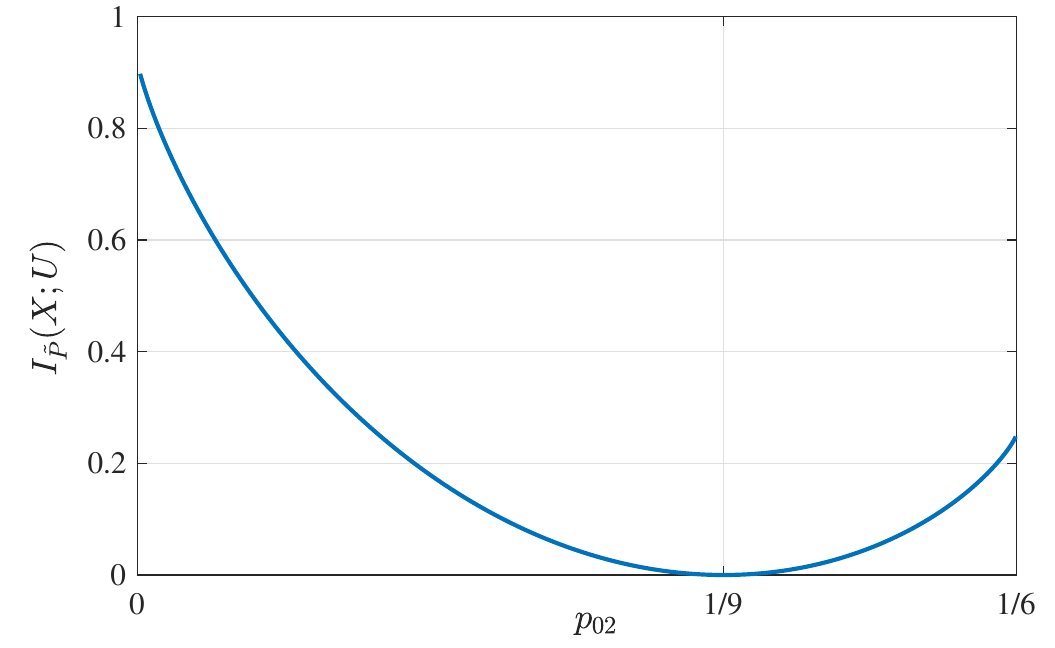}
        \caption{$I_{\Ptilde}(X;U)$ vs.~$p_{02}$ for superposition coding in the ternary source example.}
        \label{fig:inf_ux_vs_p02}
    \end{figure}
    
    Since the $d_0$-distortion equals $2(p_{01} + 2p_{02})$, we conclude from \eqref{eq:ineq1}--\eqref{eq:ineq2} that every feasible pair $(p_{01},p_{02}) \ne (p_{01,m}^*,p_{02,m}^*)$ attains a strictly higher $d_0$-distortion than $(p_{01,m}^*,p_{02,m}^*)$.  This means $(p_{01,m}^*,p_{02,m}^*)$ is also uniquely optimal for the $d_0$-minimization problem in \eqref{eq:Ptilde_ternary}, and thus, the same $d_1$-distortion is attained under mismatched encoding as that of matched encoding.
\end{proof}    

\vspace*{-1ex}
\section{Proofs for the General Alphabet Setting} \label{app:continuous}
\vspace*{-0.5ex}

We will be interested in events that hold asymptotically almost surely with respect to $\Xv$, i.e., if we view $\Xv$ as the first $n$ entries of an infinitely long i.i.d.~sequence $X_1^{\infty} = (X_1,X_2,\dotsc)$, then there exists a subset $\Xtypinf$ of such sequences such that $\PP[ X_1^{\infty} \in \Xtypinf ] = 1$, and such that the desired event holds for all sequences in $\Xtypinf$.  When this is the case, we say that the event holds {\em $\Pi_X^{\infty}$-almost surely} ($\Pi_X^{\infty}$-a.s.), or for {\em $\Pi_X^{\infty}$-almost all $\xv$}.

\subsection{Auxiliary Results}

We first provide some useful auxiliary results from the works \cite{Dem02,Kon06}, which consider ``mismatched codebooks'' in a setting with only a single distortion function $d_0$.  We start with two simple results from \cite{Dem02}.

\begin{lem} \label{lem:asymp_ball}
    {\em (Asymptotic Distortion Ball Probability \cite[Thm.~1]{Dem02})} Consider any i.i.d.~source\footnote{In \cite{Dem02} and \cite{Kon06}, more general sources beyond i.i.d.~are also considered.} $\Xv \sim \Pi_X^n$ such that $\doprod < \infty$, and for some auxiliary distribution $Q_{\Xhat}$, let $\Xvhat \sim Q_{\Xhat}^n$ be independent of $\Xv$.  Then, for any $D \in (\domin,\doprod)$, it holds for $\Pi_X^{\infty}$-almost all $\xv$ that
    \begin{equation}
        -\frac{1}{n} \log \PP\big[ d_0^n(\xv,,\Xvhat) \le nD_0 \,|\,\Xv=\xv \big] \to \Rbar^{\rm iid}(Q_{\Xhat},D_0),
    \end{equation}
    as $n \to \infty$, where $\Rbar^{\rm iid}$ is defined in \eqref{eq:R_primal}.
\end{lem}

\begin{lem} \label{lem:primal_dual}
    {\em (Primal-Dual Equivalence \cite[Thm.~2]{Dem02})} For any $\Pi_X$ and $Q_{\Xhat}$ such that $\doprod < \infty$, the expressions \eqref{eq:R_primal} and \eqref{eq:R_dual} are equal for all $D_0 \in (\domin,\doprod)$.
\end{lem}

Noting that $\Rbar^{\rm iid}(Q_{\Xhat},\cdot)$ and $\Dobar^{\rm iid}(Q_{\Xhat},\cdot)$ (see \eqref{matchedMin0}) are functional inverses of each other, the following result readily follows from Lemma \ref{lem:asymp_ball}; recall that $R_{\max}$ is defined in \eqref{eq:R_max}.

\begin{cor} \label{cor:fixed_rate}
    {\em (Achievable Matched Distortion for a Fixed Rate)} For any source $\Xv \sim \Pi_X^n$ and auxiliary distribution $Q_{\Xhat}$ such that $\doprod < \infty$, under the i.i.d.~ensemble with rate $R \in (0,R_{\max})$, it holds with probability approaching one that the codeword $\Xvhat^*$ with the smallest $d_0$-distortion satisfies $\frac{1}{n}d_0(\Xv,\Xvhat^*) \to \Dobar^{\rm iid}(Q_{\Xhat},R)$ as $n \to \infty$. 
\end{cor}

We now state a less elementary result concerning the first codeword in an infinitely long codebook whose normalized distortion is no higher than a pre-specified threshold $D_0$; note that for $D_0 \in (\domin,\doprod)$, an infinitely long codebook will always contain such a codeword. 

\begin{lem} \label{lem:empirical}
    {\em (Empirical Distribution for the First Matching Codeword \cite[Thm.~2]{Kon06})}
    Suppose that $\doprod < \infty$, and fix $D_0 \in (\domin,\doprod)$.  Consider an infinite codebook $\Cc = (\Xvhat_1,\Xvhat_2,\dotsc)$ independently drawn from $Q_{\Xhat}^n$, and let $\Phat_n$ be the empirical distribution of $(\Xv,\Xvhat^*)$ with $\Xvhat^*$ being the first codeword in $\Cc$ satisfying $d_0^n(\Xv,\Xvhat) \le n D_0$.  Then, we have the following:
    \begin{itemize}
        \item[(i)] For any measurable set $E$ (not depending on $n$) and any $\delta > 0$, it holds for $\Pi_X^{\infty}$-almost all $\xv$ that
        \begin{equation}
            \PP\big[ \big| \Phat_n(E) - \Ptilde^*_{X\Xhat}(E) \big| > \delta \,|\, \Xv = \xv \big] \to 0, \label{eq:first_cw}
        \end{equation}
        where the convergence to zero as $n \to \infty$ is exponentially fast in $n$, and $\Ptilde^*_{X\Xhat}$ is defined in \eqref{eq:Pstar_IID} with $\lambda^*$ defined according to $D_0$.
        \item[(ii)] It holds with probability one that $\Phat_n$ converges in distribution to $\Ptilde^*_{X\Xhat}$ as $n \to \infty$.
    \end{itemize}
\end{lem}

The formal statement in \cite{Kon06} actually concerns the empirical distribution of $\Xvhat$ alone, instead of the joint empirical distribution of $(\Xv,\Xvhat)$.  However, the proof therein readily gives the form stated above.  Specifically, the proof of \cite[Thm.~3]{Kon06} (a conditional limit theorem from which \cite[Thm.~2]{Kon06} follows almost immediately) shows that a certain asymptotic distribution on $\Xhat$ must equal $\Ptilde^*_{\Xhat}$ by first showing the stronger result that the asymptotic distribution on $(X,\Xhat)$ must equal $\Ptilde^*_{X\Xhat}$.  We obtain Lemma \ref{lem:empirical} by skipping this step and directly using this stronger result.

In our setting, we are not directly interested in probabilities with respect to $\Phat_n$, but instead, the average of $d_1$ with respect to $\Phat_n$ (i.e., $\frac{1}{n}d_1^n(\xv,\xvhat)$).  The convergence result in the second part of Lemma \ref{lem:empirical} implies an analogous convergence result for $\frac{1}{n}d_1^n(\xv,\xvhat)$ in the case that $d_1$ is bounded and continuous, but these requirements rule out standard distortion functions such as $d_1(x,\xhat) = (x-\xhat)^2$.  Fortunately, Lemma \ref{lem:empirical} admits a natural counterpart for averages instead of probabilities without such requirements, which we state as follows.

\begin{lem} \label{lem:empirical2}
    {\em (Empirical Averages for the First Matching Codeword \cite[Thm.~2]{Kon06})}
    Suppose that $\doprod < \infty$ and $\diprod < \infty$, and fix $D_0 \in (\domin,\doprod)$.  Consider an infinite codebook $\Cc = (\Xvhat_1,\Xvhat_2,\dotsc)$ independently drawn from $Q_{\Xhat}^n$, and let $\Phat_n$ be the empirical distribution of $(\Xv,\Xvhat^*)$ with $\Xvhat^*$ being the first codeword in $\Cc$ satisfying $d_0^n(\Xv,\Xvhat) \le n D_0$.  Then, we have the following:
    \begin{itemize}
        \item[(i)] The empirical $d_1$-distortion satisfies for any $\delta > 0$ and $\Pi_X^{\infty}$-almost all $\xv$ that
        \begin{equation}
            \PP\big[ \big| d_1(\Phat_n) - d_1(\Ptilde^*_{X\Xhat}) \big| > \delta \,|\, \Xv = \xv \big] \to 0, \label{eq:first_cw2}
        \end{equation}
        where we use the shorthand $d_1(P) = \EE_{P}[d_1(X,\Xhat)]$, $\Ptilde^*_{X\Xhat}$ is defined in \eqref{eq:Pstar_IID} with $\lambda^*$ defined according to $D_0$, and the convergence to zero as $n \to \infty$ in \eqref{eq:first_cw2} is exponentially fast in $n$.
        \item[(ii)] It holds with probability one that $d_1(\Phat_n) \to d_1(\Ptilde^*_{X\Xhat})$ as $n \to \infty$.
    \end{itemize}
\end{lem}

The proof of this result follows the exact same steps as those of Lemma \ref{lem:empirical} (given in \cite{Kon06}) based on large deviations theory.  Accordingly, we do not repeat the analysis here, but we provide the main ideas and differences in Appendix \ref{app:large_dev}.

\subsection{Proof of Lemma \ref{lem:dual_IID} (Reformulation for the i.i.d.~Ensemble)}

Recall the two equivalent expressions for $\Rbar^{\rm iid}(Q_{\Xhat},D_0)$ in \eqref{eq:R_primal} and \eqref{eq:R_dual}.  As stated following \eqref{eq:R_dual}, it is known from \cite{Dem02} that the supremum in \eqref{eq:R_dual} is uniquely attained by some $\lambda^*$ satisfying $\Lambda'(\Lambda) = D_0$ when $D_0 \in (\domin,\doprod)$, where $\Lambda(\lambda) = \EE_{\Pi}\big[ \log \EE_{Q}[e^{\lambda d_0(X,\Xhat)}] \big]$.  Moreover, it is shown in \cite{Dem02} that the unique corresponding minimizer in \eqref{eq:R_primal} is $\Ptilde^*_{X\Xhat}$, defined in \eqref{eq:Pstar_IID}.  In view of these known results, it is then natural that the formulation in Lemma \ref{lem:iid} reduces to that in Lemma \ref{lem:dual_IID}.  We now proceed by showing this formally.

While we build on \cite{Dem02}, a key issue that we need to address is that their results concern the rate after specifying $D_0$, whereas in \eqref{eq:setF_IID} we have specified $R$ but not $D_0$.  Recall that we are interested in solutions to 
\begin{equation}
    \min_{ \substack{ \Ptilde_{X\Xhat} \,:\, \Ptilde_X = \Pi_X, \\ D(\Ptilde_{X\Xhat} \| \Pi_X \times Q_{\Xhat}) \le R } } \EE_{\Ptilde}[d_0(X,\Xhat)] \label{eq:repeated}
\end{equation}
Notice that this is precisely the definition of $\Dobar^{\rm iid}(Q_{\Xhat},R)$, so the resulting minimum value is trivially $D_0 = \Dobar^{\rm iid}(Q_{\Xhat},R)$.  It remains to show that the {\em minimizer} is $\Ptilde^*_{X\Xhat}$ with $\lambda^*$ chosen according to this value of $D_0$, and that this minimizer is unique.

As outlined above, we know from \cite{Dem02} that $\Ptilde^*_{X\Xhat}$ is the unique minimizer of the problem
\begin{equation}
    \min_{ \substack{ \Ptilde_{X\Xhat} \,:\, \Ptilde_X = \Pi_X, \\ \EE_{\Ptilde}[d_0(X,\Xhat)] \le D_0} } D(\Ptilde_{X\Xhat} \| \Pi_X \times Q_{\Xhat}).  \label{eq:repeated2}
\end{equation}
Moreover, the minimum value is $R$, since the assumption $R \in (0,R_{\max})$ implies that we are in the regime where the inequality constraint in \eqref{eq:repeated} is active (this was also noted in \cite{Dem02}).  To show that $\Ptilde^*_{X\Xhat}$ also minimizes \eqref{eq:repeated}, suppose for contradiction that there were to exist a feasible distribution $\Ptilde'_{X\Xhat}$ in \eqref{eq:repeated} with value $\EE_{\Ptilde'}[d_0(X,\Xhat)] < D_0$.  Then, for $\epsilon$ sufficiently small, it must hold that the distribution $\Ptilde''_{X\Xhat} = (1-\epsilon) \Ptilde'_{X\Xhat} + \epsilon ( \Pi_X \times Q_{\Xhat} )$ satisfies $\EE_{\Ptilde''}[d_0(X,\Xhat)] \le D_0$.  But then, we have
\begin{align}
    D( \Ptilde''_{X\Xhat} \| \Pi_X \times Q_{\Xhat} )
    &= D\Big( (1-\epsilon) \Ptilde'_{X\Xhat} + \epsilon ( \Pi_X \times Q_{\Xhat} ) \big\|  (1-\epsilon) (\Pi_X \times Q_{\Xhat}) + \epsilon (\Pi_X \times Q_{\Xhat})  \Big) \\
    &\le (1-\epsilon) D( \Ptilde'_{X\Xhat} \| \Pi_X \times Q_{\Xhat} ) \le (1-\epsilon)R,
\end{align}
where we applied the convexity of KL divergence and $\epsilon D(\Pi_X \times Q_{\Xhat} \| \Pi_X \times Q_{\Xhat}) = 0$.  This means that $\Ptilde''_{X\Xhat}$ is feasible in \eqref{eq:repeated2} but yields a smaller objective than the optimal value of $R$, which is a contradiction, and establishes the desired claim.

For uniqueness, recall that we are in the regime where the inequality constraint in \eqref{eq:repeated} is active, and note that if there is another optimal distribution $\Ptilde^{\dagger}_{X\Xhat}$ distinct from $\Ptilde^*_{X\Xhat}$, then they must both yield the same values of $D( \Ptilde'_{X\Xhat} \| \Pi_X \times Q_{\Xhat} )$ and $\EE_{\Ptilde}[d_0(X,\Xhat)]$.  However, this would mean that $\Ptilde^{\dagger}_{X\Xhat}$ is also optimal in \eqref{eq:repeated2}, contradicting the uniqueness known from \cite{Dem02}.

Having established that $\Ptilde^*_{X\Xhat}$ is uniquely optimal in \eqref{eq:setF_IID}, the lemma follows.

\subsection{Proof of Theorem \ref{thm:continuous}}

We would like to use Lemma \ref{lem:empirical2} with $D_0 = \Dobar^{\rm iid}(Q_{\Xhat},R)$ to conclude that $\Phat_n$ behaves similarly to $\Ptilde_{X\Xhat}^*$, implying that the $d_1$-distortion incurred is $\EE_{\Ptilde^*}[d_1(X,\Xhat)]$.  The difficulty is that we are in the fixed-rate setting, whereas Lemma \ref{lem:empirical2} holds for the fixed-distortion setting.  Since there is a continuum of possible $D_0$ values, we cannot directly deduce that Lemma \ref{lem:empirical2} holds for the actual $d_0$-distortion value incurred by the source and the chosen codeword.  To circumvent this, we extend the almost-sure event to a countably infinite set, as well as exploiting our tie-breaking strategy that chooses the first index.


We first state the following corollary of Lemma \ref{lem:empirical2}.

\begin{cor} \label{cor:unif_d0}
    {\em (Uniform Guarantee over Rational $D_0$ Values)}
    Fix any $D_0^{\infty} \in (\domin,\doprod)$, and let $\delta > 0$ be sufficiently small such that $[D_0^{\infty} - \delta, D_0^{\infty} + \delta] \subseteq (\domin,\doprod)$.  Consider an infinite codebook $\Cc = (\Xvhat_1,\Xvhat_2,\dotsc)$ independently drawn from $Q_{\Xhat}^n$.  Let $\Ptilde^{(D_0)}_{X\Xhat}$ be as defined in \eqref{eq:Pstar_IID} with an explicit dependence on $D_0$, and let $\Phat_n^{(D_0)}$ be the empirical distribution of $(\Xv,\Xvhat^{*}_{(D_0)})$ with $\Xvhat^{*}_{(D_0)}$ being the first codeword in $\Cc$ satisfying $d_0^n(\Xv,\Xvhat) \le n D_0$.  Then, with probability one, it holds simultaneously for all rational values of $D_0 \in [D_0^{\infty} - \delta, D_0^{\infty} + \delta]$ that $d_1(\Phat_n^{(D_0)}) \to d_1(\Ptilde^{(D_0)}_{X\Xhat})$ as $n \to \infty$.
\end{cor}
\begin{proof}
    Observe that we restrict $D_0$ to ensure that $D_0 \in (\domin,\doprod)$.  For any such value of $D_0$, Lemma \ref{lem:empirical2} gives the desired convergence of $d_1(\Phat_n^{(D_0)})$ to $d_1(\Ptilde^{(D_0)}_{X\Xhat})$ almost surely.  We can then take the intersection over countably many values of $D_0$ (since we are restricting to rational numbers), and the corollary follows since a countable intersection of almost-sure events still holds almost surely.
\end{proof}

Let $D^*_{0,n} = \frac{1}{n} d_0^n(\Xv,\Xvhat^*)$ be the $d_0$-distortion incurred at block length $n$.  We know from Corollary \ref{cor:fixed_rate} that with probability $1-o(1)$, it holds that $D^*_{0,n} = D_0^{\infty} \cdot (1+o(1))$, where $D_0^{\infty} = \Dbar_0^{\rm iid}(Q_{\Xhat},R)$.  Hence, for any $\delta > 0$, we can assume that $D^*_{0,n} \in [D_0^{\infty} - \delta, D_0^{\infty} + \delta]$ for sufficiently large $n$.  In addition, we note that this choice of $D_0^{\infty}$ satisfies the condition $D_0^{\infty} \in (\domin,\doprod)$ in Corollary \ref{cor:unif_d0}, due to the fact that $R \in (0,R_{\max})$.  Thus, Corollary \ref{cor:unif_d0} reveals that $d_1(\Phat_n^{(D_0)}) \to d_1(\Ptilde^{(D_0)}_{X\Xhat})$ for all rational values of $D_0 \in [D_0^{\infty} - \delta, D_0^{\infty} + \delta]$, with probability one.

With $D^*_{0,n}$ being the minimum $d_0$-distortion induced by the codebook, we let $D^{**}_{0,n}$ be a rational number arbitrarily close to $D^*_{0,n}$, and we require it to be sufficiently close such that $D^{**}_{0,n}$ is below all the other distortion levels (not equal to $D^*_{0,n}$) induced by the codewords.  Since the rationals are dense, this is always possible.  There may potentially be multiple codewords with distortion $D^*_{0,n}$, but even if this is the case, our assumed tie breaking strategy (namely, choosing the smallest index) implies that $\Xvhat^*$ must be the first codeword having distortion at most $D^{**}_{0,n}$.  By the conclusion of the previous paragraph, it follows that the $d_1$-distortion incurred is $d_1(\Ptilde^{(D_{0,n}^{**})}_{X\Xhat}) \cdot (1+o(1))$.


To deduce the distortion level in \eqref{eq:D_iid_dual} (with $D_0$ replaced by $D_0^{\infty}$), it remains to argue that $d_1(\Ptilde^{(D_{0,n}^{**})}_{X\Xhat}) \cdot (1+o(1))$ reduces to $d_1(\Ptilde_{X\Xhat}^*) \cdot (1+o(1))$, where in the latter expression, $\Ptilde_{X\Xhat}^*$ is defined with respect to $D_0^{\infty}$.  To see this, recall that we already established $D^*_{0,n} = D_0^{\infty} \cdot (1+o(1))$ (with probability approaching one), which in turn implies that the rational rounded version can also be chosen to satisfy $D^{**}_{0,n} = D_0^{\infty} \cdot (1+o(1))$.    It is shown in \cite[Sec.~II-B]{Dem02} that the function $\Lambda(\cdot)$ defining $\lambda^*$ in \eqref{eq:Pstar_IID} satisfies $\Lambda''(\lambda) > 0$ for all $\lambda < 0$.  Thus, with $\lambda^*$ being the solution to $\Lambda'(\lambda) = D_0$ with $D_0 \in (\domin,\doprod)$, we have that $\lambda^*$ is a continuous function of $D_0$.  Thus, with $D^{**}_{0,n} \to D_0^{\infty}$, the corresponding $\lambda^*$ value similarly approaches the value corresponding to $D_0^{\infty}$.  

We now argue that the right-hand side of \eqref{eq:Pstar_IID} is continuous with respect to $\lambda^*$ for $\Pi_X$-almost all $x$ and any $\xhat$.  Observe that this ratio is non-negative and is the reciprocal of $\EE_Q[e^{\lambda^*(d_0(x,\Xhat)-d_0(x,\xhat))}]$, which is a convex function of $\lambda^*$ and is therefore continuous in the region where it is finite.  Thus, since the numerator in  \eqref{eq:Pstar_IID} is a fixed positive constant, it only remains to handle the possibility of dividing by zero or infinity.  However, we claim that the denominator $\EE_{Q}[e^{\lambda^* d_0(x,\Xhat)}]$ is at most one, and is bounded away from zero for $\Pi_X$-almost all $x$.  The former property follows from $\lambda^* \le 0$ and $d_0(\cdot,\cdot) \ge 0$, and the latter follows from $\EE_{\Pi \times Q}[e^{\lambda^* d_0(X,\Xhat)}] \ge e^{\lambda^* \doprod}$ via Jensen's inequality.  The desired continuity with respect to $\lambda^*$ readily follows.

This pointwise convergence for the Radon-Nikodym derivative implies that the distortion $d_1(\Ptilde^{(D_{0,n}^{**})}_{X\Xhat}) \cdot (1+o(1))$ reduces to $d_1(\Ptilde_{X\Xhat}^*) \cdot (1+o(1)) = \EE_{\Ptilde^*}[d_1(X,\Xhat)] \cdot (1+o(1))$, where $\Ptilde_{X\Xhat}^*$ is defined with respect to $D_0^{\infty}$ as stated above. This completes the proof of Theorem \ref{thm:continuous}; we derived the form given in \eqref{eq:D_iid_dual}--\eqref{eq:Pstar_IID}, but by Lemma \ref{lem:dual_IID}, the form \eqref{eq:D1bar_IID}--\eqref{eq:setF_IID} is equivalent.

\subsection{Discussion on Lemma \ref{lem:empirical2} (Empirical Averages for the First Matching Codeword)} \label{app:large_dev}

The proof of Lemma \ref{lem:empirical} in \cite{Kon06} is based on applying a conditional large-deviations analysis to the two-dimensional vector $\big(\frac{1}{n}d_0^n(\xv,\Xvhat), \Phat_n(E)\big)$, where $\Xvhat \sim Q_{\Xhat}^n$ and the conditioning is on the event $d_0^n(\xv,\Xvhat) \le nD_0$ (and $\Xv = \xv$).  The conditioning event $d_0^n(\xv,\Xvhat) \le nD_0$ comes from $\Xvhat$ being defined in Lemma \ref{lem:empirical} to be the first codeword satisfying this property. 

Here we outline how a similar approach can be taken for Lemma \ref{lem:empirical2}, considering the two-dimensional vector $\big(\frac{1}{n}d_0^n(\xv,\Xvhat),\frac{1}{n}d_1^n(\xv,\Xvhat)\big)$, written as $(d_0(\Phat_n),d_1(\Phat_n))$ for short.  According to the conditioning event $d_0(\Phat_n) \le D_0$ and the definition of conditional probability, we are interested in the ratio
\begin{equation}
    \frac{ \PP\big[ d_0(\Phat_n) \le D_0 \text{ and } d_1(\Phat_n) \le d_1(\Ptilde^*_{X\Xhat}) - \delta \big] }{ \PP[ d_0(\Phat_n) \le D_0 ] }, \label{eq:ld_ratio}
\end{equation}
as well as the analog with the second event replaced by $d_1(\Phat_n) \ge d_1(\Ptilde^*_{X\Xhat}) + \delta$.  The two are handled similarly, so we focus only on \eqref{eq:ld_ratio}. The denominator is a standard quantity characterized in \cite{Dem02}, so the main challenge is in upper bounding the numerator.

The relevant log-moment generating function is
\begin{equation}
    \tilde{\Lambda}_n(\blambda) = \log \EE_{Q_{\Xhat}^n}\big[ e^{\lambda_0 d_0(\Phat_n) + \lambda_1 d_1(\Phat_n)} \big],
\end{equation}
where $\blambda = (\lambda_0,\lambda_1) \in (-\infty,0]^2$, and by the ergodic theorem, $\frac{1}{n}\Lambda_n(n \blambda)$ converges almost surely to the following for $\Pi_X^{\infty}$-almost all $\xv$ (analogous to \cite[Eq.~(22)]{Kon06}):
\begin{equation}
    \tilde{\Lambda}(\blambda) = \EE_{\Pi}[ \log \EE_Q[ e^{\lambda_0 d_0(X,\Xhat) + \lambda_1 d_1(X,\Xhat)} ]].
\end{equation}
Since $d_0$ and $d_1$ are non-negative and $\lambda \in (-\infty,0]^2$, we have $\tilde{\Lambda}(\blambda) \le 0$.  On the other hand, by Jensen's inequality, $\tilde{\Lambda}(\blambda) \ge \lambda_0 \doprod + \lambda_1 \diprod$; thus, the assumptions $\doprod < \infty$ and $\diprod < \infty$ are sufficient to ensure finite $\tilde{\Lambda}(\blambda)$.  This is the only assumption made in Lemma \ref{lem:empirical2} that is absent in Lemma \ref{lem:empirical}.

Rather than giving the technical details of the large-deviations analysis, we provide the main intuition.  It is known that it weren't for the event $d_1(\Phat_n) \le d_1(\Ptilde^*_{X\Xhat}) - \delta$ in the numerator of \eqref{eq:ld_ratio}, the limiting distribution of $\Phat_n$ would be $\Ptilde^*_{X\Xhat}$ given in \eqref{eq:Pstar_IID} \cite{Dem02,Kon06}.  The idea is that {\em the event $d_1(\Phat_n) \le d_1(\Ptilde^*_{X\Xhat}) - \delta$ makes this distribution infeasible}, and accordingly, $\Phat_n$ is forced towards a slightly different distribution.  This makes the exponent of the numerator in \eqref{eq:ld_ratio} strictly larger than that of the denominator, giving the desired exponential decay to zero in the first part of Lemma \ref{lem:empirical2}.

By the same argument as that used for Lemma \ref{lem:empirical} (see \cite{Kon06}), the second part of Lemma \ref{lem:empirical2} follows easily from the first: Summing the left-hand side of \eqref{eq:first_cw2} over all $n$ gives a finite value due to the exponential decay, and then the Borel-Cantelli lemma implies that with probability one, the event $\big| d_1(\Phat_n) - d_1(\Ptilde^*_{X\Xhat}) \big| > \delta$ occurs for at most finitely many $n$.  Since $\delta$ is arbitrarily small, it follows that $d_1(\Phat_n) \to d_1(\Ptilde^*_{X\Xhat})$ almost surely, as desired.

\section{Example on the Impact of Tie-Breaking} \label{app:example_tie}
\vspace*{-0.5ex}

Consider the binary symmetric source with $\Pi_X = \big(\frac{1}{2},\frac{1}{2}\big)$, and let $Q_{\Xhat} = \big(\frac{1}{2},\frac{1}{2}\big)$.  Suppose that $d_1$ is the Hamming distortion, but $d_0(x,\xhat) = \bone\{ x = 0 \cap \xhat = 1 \}$, i.e., only $0 \to 1$ flips are penalized.   For this source with minimum $d_1$-distortion encoding (i.e., the matched case), it known that the optimal rate-distortion trade-off is given as follows (measuring information in bits):
\begin{equation}
    R = 1 - H_2(D_1) \label{eq:ex_matched}
\end{equation}
for $D_1 \in \big[0,\frac{1}{2}\big]$, and moreover, the choice $Q_{\Xhat} = \big(\frac{1}{2},\frac{1}{2}\big)$ is optimal \cite[Sec.~10.3.1]{Cov06}.

We investigate the impact of tie-breaking on the constant-composition ensemble (Lemma \ref{lem:existing}) and the i.i.d.~ensemble (Lemma \ref{lem:iid}), but instead of using the equations in the associated lemmas, we find it more convenient to use direct arguments.  We keep some steps slightly informal to convey the key ideas without being overly technical.

We start with the constant-composition ensemble.  Suppose that $\Xv$ is a typical sequence with half 0s and half 1s (the case of $\frac{1}{2} \pm \delta$ is similar).  Then, since each $\Xvhat$ is a constant-composition codeword according to $Q_{\Xhat} = \big(\frac{1}{2},\frac{1}{2}\big)$, we have that the fraction of indices where $(x,\xhat) = (0,1)$ is the same as the number of indices where $(x,\xhat) = (1,0)$.  This further implies that the minimum $d_0$-distortion rule is the same as the minimum $d_1$-distortion rule -- the latter is simply double the former.  Accordingly, the ``mismatched'' encoding is in fact matched encoding, and combining this with the optimality of $Q_{\Xhat}$ mentioned above, we conclude that we attain the matched  rate-distortion trade-off stated in \eqref{eq:ex_matched}.

For the i.i.d.~ensemble, we split the analysis into several cases:
\begin{itemize}
    \item Suppose that $R \in \big(0, \frac{1}{2})$, and consider a typical $\Xv$ sequence with half 0s and half 1s.  For the length-$\frac{n}{2}$ subsequence corresponding to zeros, a standard property of types \cite[Ch.~2]{Csi11} reveals that the probability of having an $\alpha$ fraction of 0s is roughly $2^{-\frac{n}{2}(1 - H_2(\alpha))}$.  Thus, with $2^{nR}$ codewords, the {\em smallest $\alpha$} (corresponding to minimum $d_0$-distortion encoding) yields
    \begin{equation}
        R = \frac{1}{2}(1 - H_2(\alpha)) \iff H_2(\alpha) = 1-2R.
    \end{equation}
    Moreover, we claim that the resulting $d_1$-distortion incurred is
    \begin{equation}
        D_1 = \frac{\alpha}{2} + \frac{1}{4}.
    \end{equation}
    To see this, first note that the $\frac{\alpha}{2}$ term comes directly from $0 \to 1$ flips from $\Xv$ to $\Xvhat$, with division by two since $\alpha$ is defined with respect to the length-$\frac{n}{2}$ subsequence.  Moreover, this $\alpha$-fraction event is a rare event, and with high probability there will only be few codewords (if not just one) meeting this $\alpha$-distortion level.  Accordingly, since the encoder ignores the entries where $\Xv$ equals one, the corresponding $\Xvhat$ entries will simply exhibit ``typical'' behavior, i.e., roughly half 0s and half 1s, contributing roughly $\frac{1}{4}$ to the $d_1$-distortion.
    \item Suppose that $R \in \big (\frac{1}{2},1\big)$, and again consider a typical $\Xv$ sequence with half 0s and half 1s. Since any particular length-$\frac{n}{2}$ subsequence of $\Xvhat$ has probability $2^{-\frac{n}{2}}$, the condition $R > \frac{1}{2}$ implies that with high probability, the selected codeword according to the minimum $d_0$-distortion rule will achieve $d_0^n(\Xv,\Xvhat) = 0$, i.e., all values match in the subsequence where $\Xv$ is zero.  In fact, there will be {\em exponentially many} candidates meeting this condition, and accordingly, the analysis crucially depends on the tie-breaking strategy:
    \begin{itemize}
        \item Under pessimistic tie breaking, the selected $\Xvhat$ will be the one with the {\em highest $d_1$-distortion}, i.e., having the fewest matches in the subsequence where $\Xv$ is one.  Since roughly a fraction $2^{-\frac{n}{2}}$ of the codewords achieve zero $d_0$-distortion, there are effectively $2^{n(R - \frac{1}{2})}$ sub-codewords to choose from.  By a standard property of types, the probability of one such sub-codeword having a $\beta \in \big[0,\frac{1}{2}\big]$ fraction of matches is roughly $2^{-\frac{n}{2}( 1 - H_2(\beta) )}$, so the one with the fewest matches will yield
        \begin{equation}
            R - \frac{1}{2} = \frac{1}{2}( 1 - H_2(\beta) ) \iff H_2(\beta) = 2 - 2R.
        \end{equation}
        Then, the resulting $d_1$-distortion incurred is
        \begin{equation}
            D_1 = \frac{1-\beta}{2},
        \end{equation}
        where the division by two corresponds to considering the length-$\frac{n}{2}$ subsequence where $\Xv$ equals one.
        \item Under uniformly random tie breaking, with high probability, the selected codeword will be one with typical behavior in the entries where $\Xv$ is one, so the distortion will be $\frac{1}{4}$ (i.e., roughly half flips in the relevant half of the codeword).
    \end{itemize}
\end{itemize}
The resulting rate-distortion curves are shown in Figure \ref{fig:rd_example3}.  Note that the sharp transition at $R = \frac{1}{2}$ is due to a sudden change in what dictates the $d_1$-distortion: For $R$ slightly smaller, the behavior is dictated by the fraction of matches where $\Xv$ is zero (with nearly all matches), whereas for $R$ slightly larger, it is dictated by the fraction of matches where $\Xv$ is one (with roughly half matches).

This example highlights that the pessimistic choice of tie-breaking strategy can in fact have a significant impact in certain scenarios.  We note that the change in behavior at $R = \frac{1}{2}$ corresponds exactly to the fact that $R_{\max}$ given in \eqref{eq:R_max} equals $\frac{1}{2}$ in this example.  For $R \in (0,R_{\max})$ the set $\Pctilde^{\rm iid}$ in Lemma \ref{lem:iid} is a singleton, in concordance with the uniqueness property associated with $\Ptilde^*_{X\Xhat}$ stated following Lemma \ref{lem:dual_IID}.  For $R > R_{\max}$ this is no longer the case, and the tie-breaking strategy plays a significant role.


\begin{figure}[!t]
    \centering
    \includegraphics[width=0.5\textwidth]{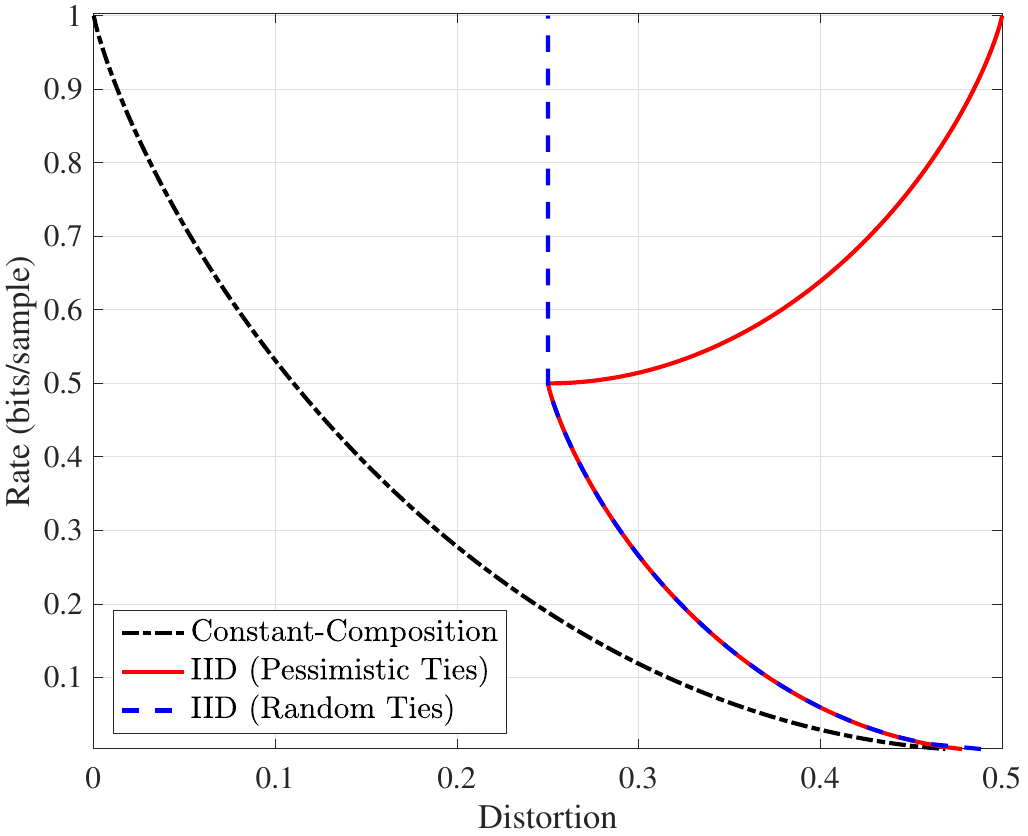}
    \par
    \caption{Mismatched rate-distortion curves for the binary symmetric source. \label{fig:rd_example3}}
    \vspace*{-3ex}
\end{figure}

\vspace*{-1ex}
\section{Continuity Argument in Achievability Analysis} \label{app:technical}
\vspace*{-0.5ex}

In this appendix, we describe the technical details of how Lemma \ref{lem:existing} follows from Lemma \ref{lem:occurence_of_types}, and discuss the analogous steps for the other ensembles that we consider.  We closely follow the steps of Lapidoth \cite{Lap97} with some minor modifications.  We assume that $R > 0$, since for $R = 0$ there is only one codeword and the analysis is trivial.

Recall the sets $\Sc_{n,\delta}^{\supseteq}$ and $\Sc_{n,\delta}^{\subseteq}$ defined in \eqref{eq:S_super}--\eqref{eq:S_sub} in Lemma \ref{lem:occurence_of_types}, where $P_{X,n}$ represents the type of $\Xv$.  By the law of large numbers, we can assume that $P_{X,n}$ converges to $\Pi_X$ as $n \to \infty$.  We proceed by implicitly conditioning on this being true, and on the high-probability events in Lemma \ref{lem:occurence_of_types} holding true.  We proceed in several steps.

{\bf Step 1.} Define
\begin{equation}
    \Sc_{\delta}^{\subseteq} = \left\{\Ptilde_{X\Xhat} \in \Pc(\Xc \times \hat{\Xc}) \,:\, \Ptilde_X = \Pi_X, \Ptilde_{\Xhat} = Q_{\Xhat}, I_{\Ptilde}(X, \Xhat) \leq R - \delta \right\}, \label{eq:S_sub2}
\end{equation}
which is analogous to $\Sc_{n,\delta}^{\subseteq}$ but with general distributions instead of types.  Moreover, let
\begin{equation}
    D^{\subseteq}_{0,\delta} = \min_{\Ptilde_{X\Xhat} \in \Sc_{\delta}^{\subseteq}} \EE_{\Ptilde}[d_0(X,\Xhat)]. \label{eq:D_sub}
\end{equation}
Defining $\Ptilde^*_{X\Xhat}$ to be the joint type induced by $\Xv$ and the selected codeword $\Xvhat$, we claim that 
\begin{equation}
    \EE_{\Ptilde^*}[d_0(X,\Xhat)] \le D^{\subseteq}_{0,\delta} + \delta,
\end{equation}
where here and subsequently, all statements are taken to mean for $n$ sufficiently large.  To see this, note from Lemma \ref{lem:occurence_of_types} that $\EE_{\Ptilde^*}[d_0(X,\Xhat)]$ is least as small as the smallest $d_0$-distortion incurred within $\Sc_{n,\delta}^{\subseteq}$, which in turn can approximate any joint distribution in $\Sc_{\delta}^{\subseteq}$ arbitrarily closely when $n$ is large enough (note also that $\EE_{\Ptilde}[d_0(X,\Xhat)]$ is continuous with respect to $\Ptilde_{X\Xhat}$ in the finite-alphabet setting, as is $I_{\Ptilde}(X;\Xhat)$).

{\bf Step 2.} 
Define
\begin{equation}
    \Sc^{\supseteq}_{\delta} = \Big\{ \Ptilde_{X\Xhat} \in \Pc(\Xc \times \Xchat) \,:\, \Ptilde_X \in B_{\delta}(\Pi_X), \Ptilde_{\Xhat} \in B_{\delta}(Q_{\Xhat}), I_{\Ptilde}(X;\Xhat) \le R + \delta \Big\}, \label{eq:S_super2}
\end{equation} 
where $B_{\delta}(P_Z)$ denotes the set of all distributions with the same support as $P_Z$ and probabilities that differ from $P_Z$ at most $\delta$ entry-wise.  Moreover, let
\begin{equation}
    D^{\supseteq}_{0,\delta} = \min_{\Ptilde_{X\Xhat} \in \Sc^{\supseteq}_{\delta}} \EE_{\Ptilde}[d_0(X,\Xhat)].
\end{equation}
The goal of this step is to show that both $D^{\subseteq}_{0,\delta}$ and $D^{\supseteq}_{0,\delta}$ approach the following value as $\delta \to 0$:
\begin{equation}
    D_0^* = \min_{ \substack{ \Ptilde_{X\Xhat} \,:\, \Ptilde_X = \Pi_X, \Ptilde_{\Xhat} = Q_{\Xhat} \\ I_{\Ptilde}(X;\Xhat) \le R  }} \EE_{\Ptilde}[d_0(X,\Xhat)]. \label{eq:D0_star}
\end{equation}
For $D^{\subseteq}_{0,\delta}$, this follows from the fact that $D_{0,\delta}^{\subseteq}$ is continuous in $\delta \in [0,R]$, which follows from  \eqref{eq:S_sub2}--\eqref{eq:D_sub} and the continuity of mutual information and $\EE_{\Ptilde}[d_0(X,\Xhat)]$.\footnote{In general, more care may be needed if $\Sc_{\delta}^{\subseteq}$ were to be empty for $\delta > 0$ but non-empty for $\delta = 0$.  However, we are focusing on the case that $R > 0$, so this does not occur (a small neighborhood of $\Pi_X \times Q_{\Xhat}$ is always included). \label{foot:caveat}}

For $D^{\supseteq}_{0,\delta}$, we form upper and lower bounds that asymptotically match.  For the upper bound, since $\cup_{\delta > 0} D^{\supseteq}_{0,\delta}$ precisely equals the constraint set in \eqref{eq:D0_star}, it immediately follows that $\limsup_{\delta \to 0} D^{\supseteq}_{0,\delta} \le D_0^*$.  For the other direction, we know that for any sequence $\{\delta_i\}_{i \ge 1}$ converging to zero, if $\Ptilde_{X\Xhat}^{(i)}$ is the corresponding sequence of minimizers achieving $\{D^{\supseteq}_{0,\delta_i}\}_{i \ge 1}$, then by the compactness of the probability simplex there must exist a convergent subsequence.  The convergence must be to a distribution satisfying the constraints in \eqref{eq:D0_star}, and it follows that $\liminf_{\delta \to 0} D^{\supseteq}_{0,\delta} \ge D_0^*$.  Combining the two limits gives $\lim_{\delta \to 0} D^{\supseteq}_{0,\delta} = D_0^*$.

{\bf Step 3.} From Step 1, Lemma \ref{lem:occurence_of_types}, and the fact that $P_{X,n} \to \Pi_X$ and $Q_{\Xhat,n} \to Q_{\Xhat}$, we know that the true induced joint type $\Ptilde^*_{X\Xhat}$ lies in the set
\begin{equation}
    \Pctilde^{\supseteq}_{\delta} = \Big\{ \Ptilde_{X\Xhat} \in \Pc(\Xc \times \Xchat) \,:\, \Ptilde_X \in B_{\delta}(\Pi_X), \Ptilde_{\Xhat} \in B_{\delta}(Q_{\Xhat}), I_{\Ptilde}(X;\Xhat) \le R + \delta, \EE_{\Ptilde}[d_0(X,\Xhat)] \le D^{\subseteq}_{0,\delta} + \delta \Big\}.
\end{equation}
Hence, the $d_1$-distortion incurred is at most $\max_{\Ptilde_{X\Xhat} \in \Pctilde^{\supseteq}_{\delta}} \EE_{\Ptilde}[ d_1(X,\Xhat) ]$.
Moreover, from Step 2, we know that $\cap_{\delta > 0} \Pctilde^{\supseteq}_{\delta} = \Pctilde$, where $\Pctilde$ is defined in \eqref{eq:setF}.  Accordingly, using compactness and continuity in the same way as Step 2, we obtain
\begin{equation}
    \lim_{\delta \to 0} \max_{\Ptilde_{X\Xhat} \in \Pctilde^{\supseteq}_{\delta}} \EE_{\Ptilde}[ d_1(X,\Xhat) ] =  \max_{\Ptilde_{X\Xhat} \in \Pctilde} \EE_{\Ptilde}[ d_1(X,\Xhat) ],
\end{equation}
which equals the desired value stated in Lemma \ref{lem:existing}.

{\bf Other ensembles.} For the i.i.d.~ensemble (Lemma \ref{lem:iid}), the constraints on $\Ptilde_{\Xhat}$ are absent, but there are otherwise no significant differences in the analogous steps to those above.  For the multi-user ensembles (Section \ref{sec:multi_user}), the only significant difference is that there are multiple mutual constraints instead of just one.  This does not significantly affect the analysis, except possibly in the argument directly after \eqref{eq:D0_star}.  In accordance with Footnote \ref{foot:caveat}, in principle we may need to be careful with $\Sc_{\delta}^{\subseteq}$ being empty whenever $\delta > 0$.  However, this will not be the case when all the rates defining the ensemble are positive, since a suitable neighborhood of the product distribution\footnote{The product distribution is $\Pi_X \times Q_{\Xhat}$ under independent codewords, $\Pi_X \times Q_{U\Xhat}$ for superposition coding, and $\Pi_X \times Q_{\Xhat_1} \times Q_{\Xhat_2}$ for expurgated parallel coding.  Under these choices, all of the relevant mutual information terms become zero.} (intersected with the marginal constraints) will always be included.  On the other hand, if any of the rates are zero (i.e., $R_0$ or $R_1$ in superposition coding, $R_1$ or $R_2$ for expurgated parallel coding), then the ensemble reduces to that of independent codewords anyway, and the distortion achieved reduces to that of Lemma \ref{lem:existing}.

For expurgated parallel coding, we also have the constraint $\Ptilde_{\Xhat_1\Xhat_2} \in B_{\delta}(Q_{\Xhat_1} \times Q_{\Xhat_2})$, which is slightly different from the equality constraints in the other ensembles.  This is easily handled in the same way to how we used $B_{\delta}(\cdot)$ in \eqref{eq:S_super2}, and in fact, Lapidoth's original analysis \cite{Lap97} was also based on an ensemble that leads to similar constraints (instead of equality constraints).

\vspace*{-1ex}
\section{Achievability Proofs for Multi-User Ensembles} \label{app:ach_proofs}
\vspace*{-0.5ex}

\subsection{Proof of Theorem \ref{thm:supcoding_achievability} (Superposition Coding)} \label{sec:PROOF_SC}

The main step of the analysis is to establish a counterpart to Lemma \ref{lem:occurence_of_types}, but now concerning the joint types $\Ptilde_{XU\Xhat}$ induced by triplets $(\xv,\Uv^{(i)},\Xvhat^{(i,j)})$, with $i=1,\dotsc,M_0$ and $j=1,\dotsc,M_1$.  Throughout the bulk of the analysis, we condition on a fixed $\Xv = \xv$, whose type we denote by $P_{X,n}$.

The relevant marginal constraints follow immediately by construction: Any joint type $\Ptilde_{XU\Xhat}$ that occurs must satisfy $\Ptilde_X = P_{X,n}$ and $\Ptilde_{U\Xhat} = Q_{U\Xhat,n}$ (see \eqref{eq:supcodingXhat}).  The non-trivial part is to establish that all such joint types with $I_{\Ptilde}(X;U) \leq R_0 - \delta$ and $I_{\Ptilde}(X; U, \Xhat) \leq R_0 + R_1 - \delta$ occur, and all such joint types with $I_{\Ptilde}(X;U) \geq R_0 +\delta$ or $I_{\Ptilde}(X; U, \Xhat) \geq R_0 + R_1 + \delta$ do not, where $\delta > 0$ is arbitrarily small.  

To do so, we fix a joint type $\Ptilde_{XU\Xhat}$ satisfying the above marginal constraints, and consider the probability
\begin{equation}
    P_{\rm existence} = \PP\bigg[ \bigcup_{i,j} \Big\{ (\xv,\Uv^{(i)},\Xvhat^{(i,j)}) \in \Tc^n(\Ptilde_{XU\Xhat}) \Big\} \bigg].
\end{equation}
To lighten notation, we define the following events:
\begin{gather}
    \Ec_i = \bigcup_{j}\big\{ (\xv,\Uv^{(i)},\Xvhat^{(i,j)})  \in \Tc^n(\Ptilde_{XU\Xhat}) \big\}, \\
    \Ec_{ij} = \{ (\xv,\Uv^{(i)},\Xvhat^{(i,j)})  \in \Tc^n(\Ptilde_{XU\Xhat}) \}.
\end{gather}
By separating the unions over $i$ and $j$, we obtain
\begin{equation}
    P_{\rm existence} = \PP\bigg[ \bigcup_i \Ec_i \bigg] = 1-(1-\PP[\Ec_1])^{M_0}, \label{eq:p_exist_SC}
\end{equation}
where we used the fact that the events $\Ec_1,\dotsc,\Ec_{M_0}$ are independent under the superposition codebook distribution  Moreover, we can characterize $\PP[\Ec_1]$ by writing
\begin{equation}
    \PP[\Ec_1] = \PP\big[ (\xv,\Uv) \in \Tc^n(\Ptilde_{XU}) \big] \PP\bigg[ \bigcup_{j} \Big\{  (\xv,\uv,\Xvhat^{(j)}) \in \Tc^n(\Ptilde_{XU\Xhat}) \Big\} \bigg],
\end{equation}
where in the first probability, $\Uv$ is a shorthand for $\Uv^{(i)}$, and in the second probability we implicitly condition on an arbitrary fixed realization $\Uv = \uv$, and write $\Xv^{(j)}$ as a shorthand for $\Xvhat^{(i,j)}$.  Due to the conditioning on $(\xv,\uv)$, the union over $j$ is now a union of independent events, meaning that the truncated union bound is tight to within a factor of $\frac{1}{2}$ \cite{Shu03}, i.e., 
\begin{equation}
    \PP[\Ec_1] = \PP\big[ (\xv,\Uv) \in \Tc^n(\Ptilde_{XU}) \big] \times \alpha \min\{1, M_1 \PP[ (\xv,\uv,\Xvhat \in \Tc^n(\Ptilde_{XU\Xhat}) ]\} \label{eq:Pr_E1}
\end{equation}
for some $\alpha \in \big[\frac{1}{2},1\big]$, where $\Xvhat = \Xvhat^{(1)}$.  By standard properties of types \cite[Ch.~2]{Csi11}, the probabilities in this expression behave as
\begin{gather}
    \PP\big[ (\xv,\Uv) \in \Tc^n(\Ptilde_{XU}) \big] = e^{-nI(U;X) + o(n)}, \label{eq:types1} \\
    \PP[ (\xv,\uv,\Xvhat \in \Tc^n(\Ptilde_{XU\Xhat}) ] = e^{-nI(X;\Xhat|U)  + o(n)}.  \label{eq:types2} 
\end{gather}

To prove the non-existence claim, we first apply the union bound to \eqref{eq:p_exist_SC} to obtain $P_{\rm existence} \le M_0 \PP[\Ec_1]$.  If $I_{\Ptilde}(X;U) \ge R_0 + \delta$, then by upper bounding the minimum in \eqref{eq:Pr_E1} by one and using \eqref{eq:types1}, we get that $P_{\rm existence} \to 0$ exponentially fast.  Similarly, if $I_{\Ptilde}(X; U, \Xhat) \geq R_0 + R_1 + \delta$, then we upper bound the minimum in \eqref{eq:Pr_E1} by the second term, and apply \eqref{eq:types1}--\eqref{eq:types2} and the chain rule for mutual information to get $\PP[\Ec_1] \le M_1 e^{-n I(X;U,\Xhat)  + o(n)}$, which again yields $P_{\rm existence} \to 0$ exponentially fast.  Due to the exponentially fast decay, we can safely take a union bound over all (polynomially many) joint types to establish the desired non-existence result. 

To prove the existence claim, we consider two cases depending on whether the minimum in \eqref{eq:Pr_E1} is achieved by the first or second term.  If the minimum equals one, then we combine \eqref{eq:types1} with \eqref{eq:p_exist_SC} to obtain
\begin{equation}
    P_{\rm existence} \ge 1 - (1 - \alpha e^{-nI(U;X)  + o(n)})^{M_0}.
\end{equation}
Since $M_0 = e^{nR_0}$ with $R_0 \ge I_{\Ptilde}(X;U) - \delta$, it follows that $P_{\rm existence} \to 1$ faster than exponentially.\footnote{Given $b > a > 0$, note that $(1-e^{-an})^{e^{bn}} = \big( (1-e^{-an})^{e^{an}} \big)^{e^{(b-a)n}}$, and then use the fact that $(1-e^{-an})^{e^{an}}$ approaches $\frac{1}{e} \in (0,1)$.}  By a similar argument, when the minimum in \eqref{eq:Pr_E1} equals the second term, we use $R_0 + R_1 \ge I_{\Ptilde}(X;U,\Xhat) - 2\delta$, and we combine \eqref{eq:types1}--\eqref{eq:types2} with \eqref{eq:p_exist_SC} to obtain $P_{\rm existence} \to 1$ faster than exponentially.  By a union bound over all of the relevant joint types, we deduce that they all must occur in the codebook with high probability, as desired.

Having established the existence and non-existence of types within ``inner'' and ``outer'' sets that nearly match (in analogy with Lemma \ref{lem:occurence_of_types}), Theorem \ref{thm:supcoding_achievability} now follows from a similar continuity argument to the case of independent codewords; see Appendix \ref{app:technical} for details.

\subsection{Proof of Theorem \ref{thm:parcoding_achievability} (Expurgated Parallel Coding)} \label{sec:PROOF_EX}

We make use of the following lower bound on the probability of a union.

\begin{lem}[de Caen's bound \cite{Dec97}] \label{lem:decaen}
    For any finite sequence of events $\Ac_1, \dots, \Ac_N$ on a probability space, we have
    \begin{align}
        \mathbb{P}\left[\bigcup_{l=1}^N \Ac_l\right] \geq \sum_{l=1}^{N}\frac{\mathbb{P}[\Ac_l]^2}{\sum_{l'=1}^N \mathbb{P}[\Ac_l \cap \Ac_{l'}]}.
    \end{align}
\end{lem}

We condition on an arbitrary fixed realization of $\xv$, and consider joint types $\Ptilde_{X\Xhat_1 \Xhat_2}$ whose $X$-marginal matches the type of $\xv$, and whose other marginals coincide with $Q_{\Xhat_1}$ and  $Q_{\Xhat_2}$.
Given such a joint type $ \Ptilde_{X\Xhat_1, \Xhat_2} $, we are interested in the following existence probability: 
\begin{gather}
    P_{\text{existence}} = \PP\bigg[ \bigcup_{i=1}^{M_1}\bigcup_{j=1}^{M_2}\Ec_{ij}(\Ptilde_{X\Xhat_1, \Xhat_2}) \bigg], \\
    \Ec_{ij}(\Ptilde_{X\Xhat_1, \Xhat_2}) = \{(\xv, \Xvhat_1^{(i)}, \Xvhat_2^{(j)}) \in \Tc^n(\Ptilde_{X\Xhat_1\Xhat_2})\},
\end{gather}
where we again implicitly condition on $\Xv = \xv$.  The initial steps follow those of the analogous ensemble for channel coding \cite{Sca16a}. Considering the probability of the intersection $\mathbb{P}[\Ec_{ij}(\Ptilde_{X\Xhat_1\Xhat_2}) \cap \Ec_{i'j'}(\Ptilde_{X\Xhat_1\Xhat_2})]$, we have the following four cases:
\begin{enumerate}
    \item ($ i=i', j=j' $) In this case, the intersection is just the event itself, and we have
    \begin{align}
        \Psi_{00} \triangleq \mathbb{P}[(\xv, \Xvhat_1, \Xvhat_2) \in \Tc^n(\Ptilde_{X\Xhat_1\Xhat_2z}) ], \label{eq:psi00}
    \end{align}
    where $(\Xvhat_1, \Xvhat_2)$ denotes an arbitrary fixed $(\Xvhat_1^{(i)}, \Xvhat_2^{(j)})$ pair. 
    \item ($ i\neq i', j=j' $) In this case, the probability of the intersection is given by
    \begin{align}
        \Psi_{01} \triangleq \mathbb{P}[(\xv, \Xvhat_2) \in \Tc^n(\Ptilde_{X\Xhat_2})]\cdot\mathbb{P}[(\xv, \Xvhat_1, \xvhat_2) \in \Tc^n(\Ptilde_{X\Xhat_1\Xhat_2}) ]^2, \label{eq:psi01} 
    \end{align}
    where $\xvhat_2$ is an arbitrary fixed sequence such that $(\xv, \xvhat_2) \in \Tc^n(\Ptilde_{X\Xhat_2})$.
    \item ($ i = i', j \neq j' $) In this case, the probability of the intersection is given by
    \begin{align}
        \Psi_{10} \triangleq \mathbb{P}[(\xv, \Xvhat_1) \in \Tc^n(\Ptilde_{X\Xhat_1})]\cdot\mathbb{P}[(\xv, \xvhat_1, \Xvhat_2) \in \Tc^n(\Ptilde_{X\Xhat_1\Xhat_2}) ]^2, \label{eq:psi10}
    \end{align}
    where $\xvhat_1$ is an arbitrary fixed sequence such that $(\xv, \xvhat_1) \in \Tc^n(\Ptilde_{X\Xhat_1})$.
    \item ($ i \neq i', j \neq j' $) In this case, the probability of the intersection is given by $ \Psi_{11} = \Psi_{00}^2 $, due to independence.
\end{enumerate}
From Lemma \ref{lem:decaen}, identifying $l$ with $(i, j)$, $N$ with $M_1M_2$, and $\Ec_{ij}(\Ptilde_{X\Xhat_1\Xhat_2})$ with $\Ac_l$, we have
\begin{align} 
    P_{\text{existence}} &\geq \sum_{i=1}^{M_1}\sum_{j=1}^{M_2} \frac{\Psi_{00}^2}{\Psi_{00} + \sum_{i=1}^{M_1}\Psi_{10} + \sum_{j=1}^{M_2}\Psi_{01} + \sum_{i=1}^{M_1}\sum_{j=1}^{M_2}\Psi_{11}}\\
    &= \frac{M_1M_2\Psi_{00}^2}{\Psi_{00} + M_1\Psi_{10} + M_2\Psi_{01} + M_1M_2\Psi_{00}^2}\\
    &= \frac{1}{\frac{1}{M_1M_2\Psi_{00}} + \frac{\Psi_{10}}{M_2\Psi_{00}^2} + \frac{\Psi_{01}}{M_1\Psi_{00}^2} + 1} \label{eq:refUnionBound_1}
\end{align}
We now observe from \eqref{eq:psi00}--\eqref{eq:psi10} that
\begin{align}
    \Psi_{00}^2 &= \mathbb{P}[(\xv, \Xvhat_2) \in \Tc^n(\Ptilde_{X\Xhat_2})]\cdot \Psi_{01}\\
    &= \mathbb{P}[(\xv, \Xvhat_1) \in \Tc^n(\Ptilde_{X\Xhat_1})]\cdot\Psi_{10},
\end{align}
which implies that \eqref{eq:refUnionBound_1} can be rewritten as
\begin{align}
    P_{\text{existence}} &\geq \frac{1}{\frac{1}{M_1M_2\mathbb{P}[(\xv, \Xvhat_1, \Xvhat_2) \in \Tc^n(\Ptilde_{X\Xhat_1\Xhat_2})]} + \frac{1}{M_2\mathbb{P}[(\xv, \Xvhat_2) \in \Tc^n(\Ptilde_{X\Xhat_2})]} + \frac{1}{M_1\mathbb{P}[(\xv, \Xvhat_1) \in \Tc^n(\Ptilde_{X\Xhat_1})]} + 1}. \label{eq:exist_lb}
\end{align}

Fix an arbitrarily small constant $\delta > 0$, and suppose that $ I_{\Ptilde}(X;X_1) \leq R_1-\delta$, $ I_{\Ptilde}(X;X_2) \leq R_2-\delta$ and $ I_{\Ptilde}(X;X_1, X_2) \leq R_1 + R_2 -\delta$.  For independently chosen $\Xvhat_1$ and $\Xvhat_2$, the probability of inducing the joint type $\Ptilde_{X\Xhat_{\nu}}$ behaves as $e^{-nI_{\Ptilde}(X; X_{\nu})  + o(n)}$ and similarly, the probability of inducing $\Ptilde_{X\Xhat_1\Xhat_2}$ behaves as $e^{-nI_{\Ptilde}(X; \Xhat_1, \Xhat_2)  + o(n)}$ \cite{Sca16a}.  As a result, we deduce that $P_{\text{existence}}$ tends to one exponentially fast as $n$ grows large.  By a union bound over all of the relevant joint types, we deduce from \eqref{eq:exist_lb} that they all must occur in the codebook (before expurgation) with high probability, as desired.

For the non-existence part, the argument is simpler, only requiring standard upper union bounds.  Specifically, we can apply a union bound over all $(i,j)$ to get an upper bound of $M_1M_2\mathbb{P}[(\xv, \Xvhat_1, \Xvhat_2) \in \Tc^n(\Ptilde_{X\Xhat_1\Xhat_2})]$.  We can also upper bound the existence probability for triplets $(\xv, \Xvhat_1, \mywidehat{\Xv}_2)$ by the existence probability for pairs, i.e., $(\xv, \mywidehat{\Xv}_1)$ or $(\xv, \mywidehat{\Xv}_2)$.  Thus, $M_1\mathbb{P}[(\xv, \mywidehat{\Xv}_1) \in \Tc^n(\Ptilde_{X\Xhat_1})]$ and $M_2\mathbb{P}[(\xv, \mywidehat{\Xv}_2) \in \Tc^n(\Ptilde_{X\Xhat_2})]$ are also upper bounds on $P_{\rm existence}$, and the desired non-existence claim readily follows similarly to the superposition coding ensemble.

We also need to consider the effect of the expurgation step in the ensemble.  Trivially, after expurgation, for any joint type having an $(\Xhat_1,\Xhat_2)$ marginal whose $\ell_{\infty}$-norm to $Q_{\Xhat_1} \times Q_{\Xhat_2}$ exceeds $\delta$, the joint type will not appear after expurgation.  In contrast, if this $\ell_{\infty}$-norm is $\delta$ or smaller and the joint type occurred before expurgation, then it will also occur after expurgation.  Hence, expurgation simply imposes the additional constraint $|\Ptilde_{\Xhat_1\Xhat_2}(\xhat_1,\xhat_2) -  Q_{\Xhat_1}(\xhat_1)Q_{\Xhat_2}(\xhat_2)| \le \delta$ for all $(\xhat_1,\xhat_2)$.  

Finally, we note that with probability approaching one, expurgation does not impact the asymptotic rate.  This is because by the law of large numbers, the probability of any specific $(\Xvhat_1^{(i)},\Xvhat_2^{(j)})$ being expurgated tends to zero.  Hence, by Markov's inequality, the probability of more than half the pairs being expurgated approaches zero.

Having established the existence and non-existence of types with near-matching ``inner'' and ``outer'' sets (in analogy with Lemma \ref{lem:occurence_of_types}), Theorem \ref{thm:supcoding_achievability} now follows from a similar continuity argument to the case of independent codewords; see Appendix \ref{app:technical} for details.

\section*{Acknowledgment}

We are very grateful to Ioannis Kontoyiannis for bringing our attention to the paper \cite{Kon06} and providing us with helpful comments.  This work was supported by an NUS Early Career Research Award.

\bibliographystyle{IEEEtran}
\bibliography{refsMerged}

\end{document}